\documentclass[a4paper,USenglish,cleveref,thm-restate]{lipics-v2021}
\usepackage{tabularray}
\usepackage{booktabs}
\usepackage[linesnumbered]{algorithm2e}
\usepackage{tikz}
\usepackage{forest}
\usepackage{float}
\usetikzlibrary{arrows.meta,calc,fit,positioning,decorations.pathreplacing}
\usetikzlibrary{overlay-beamer-styles}

\newcommand{\RMin}{\mathsf{RMin}}
\newcommand{\RMax}{\mathsf{RMax}}
\newcommand{\PSV}{\mathsf{PSV}}
\newcommand{\NSV}{\mathsf{NSV}}
\newcommand{\PLV}{\mathsf{PLV}}
\newcommand{\NLV}{\mathsf{NLV}}
\newcommand{\Min}{\mathsf{Min}}
\newcommand{\Max}{\mathsf{Max}}
\newcommand{\cMin}{\mathsf{cMin}}
\newcommand{\cMax}{\mathsf{cMax}}
\newcommand{\rev}{\operatorname{rev}}
\newcommand{\Perm}{\mathcal P}
\newcommand{\Aneq}{\mathcal A'}
\newcommand{\Agen}{\mathcal A}
\newcommand{\BP}{\mathsf{BP}}
\newcommand{\clr}{\mathsf{color}}
\newcommand{\PRS}{\mathsf{PRS}}
\newcommand{\NRS}{\mathsf{NRS}}
\newcommand{\excess}{\mathsf{excess}}
\newcommand{\wordacc}{\mathsf{word\_access}}
\newcommand{\rank}{\mathsf{rank}}
\newcommand{\select}{\mathsf{select}}
\newcommand{\Xmin}{X_{\min}}
\newcommand{\Xmax}{X_{\max}}



\title{Optimal Encodings for Next/Previous Larger/Smaller Value Queries: A Unified Approach}

\author{Seungbum Jo}{Chungnam National University, South Korea}{sbjo@cnu.ac.kr}{https://orcid.org/0000-0002-8644-3691}{}
\author{Geunho Kim}{The University of Tokyo, Japan}{research@gnho.kim}{https://orcid.org/0000-0003-4705-0489}{}
\author{Kunihiko Sadakane}{The University of Tokyo, Japan}{sada@mist.i.u-tokyo.ac.jp}{https://orcid.org/0000-0002-8212-3682}{}
\author{Srinivasa Rao Satti}{Norwegian University of Science and Technology, Norway}{srinivasa.r.satti@ntnu.no}{https://orcid.org/0000-0003-0636-9880}{}

\hideLIPIcs
\begin{document}
\nolinenumbers
\authorrunning{S. Jo, G. Kim, K. Sadakane, and S. R. Satti}

\maketitle
\pagenumbering{gobble}

\begin{abstract}
Given a 1D array over a total order and a query position in the array, a \textit{next larger value query} returns the nearest position to the right whose value is larger than the value at the query position. Next smaller value queries and previous larger/smaller value queries are defined analogously. These queries are closely related to the $k$-th \textit{range minimum/maximum problem}, which asks for the position of the $k$-th leftmost minimum/maximum value when the query range contains multiple minimum/maximum values.

In this paper, we consider the \textit{encoding} problem for answering any arbitrary subfamily of these four queries, where the queries must be answered without accessing the input array. We give a unified approach for obtaining space lower bounds for such encodings. We also obtain data structures for the query subfamilies that match the lower bounds up to lower-order additive terms, while supporting the queries in constant time. Our results generalize and improve several previous results on this problem [Gawrychowski and Nicholson, ICALP-2015; Jo and Satti, TCS-2016; Tsur, IPL-2019; Jo and Kim, Algorithmica-2025].

In addition, these results give tight bounds on the encoding complexity of all subfamilies of these queries for (1) permutations, (2) arrays without consecutive equal elements, and (3) general arrays.
\end{abstract}

\newpage
\pagenumbering{arabic}

\section{Introduction}

Given an array $A[1,\dots,n]$ over a total order, and a query position $i$, we consider the following queries:
\begin{enumerate}
    \item $\PSV(i)$, the \emph{previous smaller value query}, returns
    $\max\{j<i : A[j] < A[i]\}$, or $0$ if no such position exists.

    \item $\PLV(i)$, the \emph{previous larger value query}, returns
    $\max\{j<i : A[j] > A[i]\}$, or $0$ if no such position exists.

    \item $\NSV(i)$, the \emph{next smaller value query}, returns
    $\min\{j>i : A[j] < A[i]\}$, or $n+1$ if no such position exists.

    \item $\NLV(i)$, the \emph{next larger value query}, returns
    $\min\{j>i : A[j] > A[i]\}$, or $n+1$ if no such position exists.
\end{enumerate}

We also consider the following range queries. Given an interval $[i,j]$ and a positive integer $k$, let $p_1 < p_2 < \cdots < p_t$ be the positions of the minimum element in $A[i,\dots,j]$, and let $q_1 < q_2 < \cdots < q_s$ be the positions of the maximum element in $A[i,\dots,j]$.
\begin{enumerate}
    \item The \emph{$k$-th range minimum query}, denoted by $\RMin(i,j,k)$, returns $p_k$ if $k\le t$, and $p_t$ otherwise.
    \item The \emph{$k$-th range maximum query}, denoted by $\RMax(i,j,k)$, returns $q_k$ if $k\le s$, and $q_s$ otherwise.
\end{enumerate}

We consider the above queries on the following classes of arrays:
\begin{enumerate}
    \item $\Perm = \bigcup_{n\geq 1}\Perm_n$, where $\Perm_n$ denotes the set of all permutations of length $n$.

    \item $\Aneq = \bigcup_{n\geq 1}\Aneq_n$, where $\Aneq_n$ denotes the set of all arrays $A[1,\dots,n]$ such that $A[i-1]\ne A[i]$ for every $1< i \le n$.
    
    \item $\Agen = \bigcup_{n\geq 1}\Agen_n$, where $\Agen_n$ denotes the set of all arrays of length $n$.
\end{enumerate}

These queries are widely studied primitives for accessing the order structure of an array. Previous/next smaller value queries are used to support navigational queries on compressed suffix trees~\cite{DBLP:conf/spire/OhlebuschFG10}. They also arise in the all-nearest-smaller-values problem, which has been studied in the context of parallel algorithms~\cite{DBLP:journals/jal/BerkmanSV93}. 
Also, as a special case of $k$-th range minimum queries, the problem to find the median among them (range minima in the middle queries) was considered by Fischer and Heun~\cite{fischer2010} for efficient construction of compressed suffix arrays.

We study the problem of designing data structures for each subfamily of $\{\PSV,\PLV,\NSV,\NLV\}$ on each of the three array classes.
We work in the \emph{encoding model}~\cite{raman2015}: after preprocessing, the input array is discarded, and all queries must be answered using only the stored representation.
In this model, if a query family distinguishes $N$ equivalence classes of input objects, then any representation requires at least $\lceil\log N\rceil$ bits\footnote{All logarithms are base two. Throughout the paper, we omit ceilings and floors when they affect the bounds only by an additive constant.}; we call this the \emph{information-theoretic lower bound} (ITLB).
Here, two objects are equivalent if the queries in the family give identical answers on them. A representation matching the ITLB up to lower-order terms is called \emph{succinct}.
In many cases, the ITLB is much smaller than the space needed to store the input array itself~\cite{raman2015}.
Throughout the paper, we assume the word-RAM model with word size $w=\Theta(\log n)$.

We also consider range queries because of their close connections to the previous/next smaller/larger value queries.
From a lower bound perspective, $\{\PSV,\NSV\}$ and $\{\RMin\}$ induce the same equivalence classes of arrays, as do $\{\PLV,\NLV\}$ and $\{\RMax\}$, and hence have the same ITLBs (see Proposition~\ref{prop:query-eq}).
For upper bounds, \emph{colored 2d-min heaps}~\cite{fischer2011} provide a structure for supporting $\PSV$, $\NSV$, and $\RMin$ simultaneously; \emph{colored 2d-max heaps} similarly support $\PLV$, $\NLV$, and $\RMax$ (see Section~\ref{sec:pre} for definitions).

\subsection{Previous results}

For permutations, every range has a unique minimum. It is known that two permutations have identical answers to all range minimum queries if and only if their Cartesian trees are identical. Since the number of Cartesian trees on $n$ nodes is the $n$-th Catalan number (\href{https://oeis.org/A000108}{A000108})\footnote{OEIS identifiers refer to entries in the On-Line Encyclopedia of Integer Sequences~\cite{oeis}}~\cite{guerrini2017}, the ITLB for supporting range minimum queries on permutations is $2n-\Theta(\log n)$ bits.

Fischer and Heun~\cite{fischerheun2011} matched this bound with a $2n+o(n)$-bit data structure supporting $\PSV$ and range minimum queries in constant time. Their representation is based on the 2d-min heap, an alternative representation of the Cartesian tree; see \Cref{sec:pre} for the definition. The same data structure also applies to general arrays with the same space bound; in the presence of equal minima, it returns the rightmost minimum in the query range, which corresponds to $\RMin(\cdot,\cdot,n)$ in our notation.

Fischer and Heun~\cite{fischer2010} also considered the problem of finding a minimum whose rank among all minima in the query range lies in the middle (range minima in the middle queries).
For general arrays, they gave a data structure using $n\log(3+2\sqrt{2})+o(n)\approx 2.54n+o(n)$ additional bits that returns, in constant time, a minimum position whose rank among the $t$ minimum positions in the query range lies between $t/16$ and $15t/16$.

Fischer~\cite{fischer2011} studied the simultaneous support of $\PSV$ and $\NSV$. For permutations, $\NSV$ can be supported in constant time using the data structure of Fischer and Heun~\cite{fischerheun2011}. For general arrays, Fischer introduced colored 2d-min heaps and gave a data structure using $n\log(3+2\sqrt{2}) +o(n) \approx 2.54n+o(n)$ bits that supports $\PSV$ and $\NSV$ in constant time. He also proved a matching lower bound based on the little Schr\"oder numbers. Later, Jo and Satti~\cite{josatti2016} observed that this representation can also support $\RMin$ for arbitrary $k$ in constant time.

Gawrychowski and Nicholson~\cite{gawrychowski2015} considered the simultaneous support of $\PSV$ and $\PLV$ on permutations.
Their representation combines the 2d-min and 2d-max heaps and supports range minimum and maximum queries in constant time using $3n+o(n)$ bits.
They also showed that this representation is succinct by proving that a Baxter permutation can be reconstructed from its answers to $\PSV$ and $\PLV$ queries, yielding an ITLB of $3n-\Theta(\log n)$ bits.\footnote{The same lower bound is also implied by the bijection of Dulucq and Guibert~\cite{dulucq1996stack}, which shows that pairs of minimum and maximum Cartesian trees are counted by the Baxter numbers (\href{https://oeis.org/A001181}{A001181}).}
Jo and Satti~\cite{josatti2016} showed that the same bound can be achieved for arrays in $\Aneq$.
For general arrays, they also gave a representation using $n\log 9 + o(n) \approx 3.17n+o(n)$ bits with $O(n)$ query time.

For the full query family $\{\PSV,\PLV,\NSV,\NLV\}$ together with $\RMin$ and $\RMax$, Jo and Satti~\cite{josatti2016} gave a data structure for $\Aneq$ using $4n+o(n)$ bits with constant query time.
For $\Agen$, they gave two data structures: one using $n \log 17 + o(n) \approx 4.09n+o(n)$ bits with linear query time and another using $(3+\log 3)n +o(n) \approx 4.58n+o(n)$ bits with constant query time.
Tsur~\cite{tsur2019} improved the space to $(2+\log 3)n \approx 3.58n$ bits for $\Aneq$ and $n\log 13 \approx 3.70n$ bits for $\Agen$, but with linear query time.
More recently, Jo and Kim~\cite{jokim2025} obtained the same space bounds up to lower-order terms and improved the query time to $O(\log^{(\ell)} n)$ for any constant $\ell$, where $\log^{(1)} n=\log n$ and $\log^{(i)} n=\log(\log^{(i-1)} n)$.
They also established a lower bound of $n\log 9-\Theta(\log n)\approx 3.17n-\Theta(\log n)$ bits for supporting $\{\RMin,\RMax\}$ on arrays in $\Agen$.

\subsection{Our results}

We obtain the first ITLBs and succinct data structures for the full query family on $\Aneq$ and $\Agen$, resolving the open problem of constructing succinct representations for this family~\cite{josatti2016,jokim2025}.
These results are stated in the following two theorems.

\begin{restatable}[Lower bounds]{theorem}{restAllLb}\label{thm:all-lb}
The ITLB for simultaneously supporting $\PSV$, $\PLV$, $\NSV$, and $\NLV$ is
\begin{enumerate}
    \item $n\log(6+4\sqrt{2})-\Theta(\log n)\approx 3.54n-\Theta(\log n)$ bits for arrays in $\Aneq_n$;
    \item $n\log(7+4\sqrt{2})-\Theta(\log n)\approx 3.66n-\Theta(\log n)$ bits for arrays in $\Agen_n$.
\end{enumerate}
The same ITLBs hold for simultaneously supporting $\RMin$ and $\RMax$ on the respective array classes.
\end{restatable}

\begin{restatable}[Upper bounds]{theorem}{restAllUb}\label{thm:all-ub}
There exist data structures that simultaneously support $\PSV$, $\PLV$, $\NSV$, $\NLV$, $\RMin$, and $\RMax$ in constant time using
\begin{enumerate}
    \item $n\log(6+4\sqrt{2})+o(n)\approx 3.54n+o(n)$ bits for arrays in $\Aneq_n$;
    \item $n\log(7+4\sqrt{2})+o(n)\approx 3.66n+o(n)$ bits for arrays in $\Agen_n$.
\end{enumerate}
\end{restatable}

\begin{table}[t]
    \centering
    \resizebox{0.8\textwidth}{!}{
        \begin{tabular}{crcc}
            \toprule
            Array type & \multicolumn{1}{c}{Space (in bits)} & Query time & Reference \\
            \midrule
            \multicolumn{4}{c}{Upper bounds} \\
            \midrule
            $\Perm$ & $3n+o(n)$ & $O(1)$ &
            \cite{fischer2011, gawrychowski2015, josatti2016} \\
            \midrule
            \multirow[c]{4}{*}{$\Aneq$} & $4n+o(n)$ & $O(1)$ &
            \cite{josatti2016} \\
            & $\begin{aligned}
                (2+\log 3)n + o(n)
                \approx 3.58n + o(n)
            \end{aligned}$ & $O(n)$ & \cite{tsur2019} \\
            & $\begin{aligned}
                (2+\log 3)n+o(n)
                \approx 3.58n+o(n)
            \end{aligned}$ & $O(\log^{(\ell)} n)$ & \cite{jokim2025} \\
            & $\begin{aligned}
                n\log(6+4\sqrt{2})+o(n)
                \approx 3.54n+o(n)
            \end{aligned}$ & $O(1)$ & Theorem~\ref{thm:all-ub} \\
            \midrule
            \multirow[c]{5}{*}{$\Agen$} & $\begin{aligned}
                (3 + \log 3)n+o(n)
                \approx 4.58n+o(n)
            \end{aligned}$ & $O(1)$ & \cite{josatti2016} \\
            & $\begin{aligned}
                n\log 17+o(n)
                \approx 4.09n+o(n)
            \end{aligned}$ & $O(n)$ & \cite{josatti2016} \\
            & $\begin{aligned}
                n\log 13 + o(n)
                \approx 3.70n + o(n)
            \end{aligned}$ & $O(n)$ & \cite{tsur2019} \\
            & $\begin{aligned}
                n\log 13+o(n)
                \approx 3.70n+o(n)
            \end{aligned}$ & $O(\log^{(\ell)} n)$ & \cite{jokim2025} \\
            & $\begin{aligned}
                n\log(7+4\sqrt{2})+o(n)
                \approx 3.66n+o(n)
            \end{aligned}$ & $O(1)$ & Theorem~\ref{thm:all-ub} \\
            \midrule
            \multicolumn{4}{c}{Lower bounds} \\
            \midrule
            $\Perm$ & $3n-\Theta(\log n)$ & &
            \cite{dulucq1996stack, gawrychowski2015} \\
            \midrule
            $\Aneq$ & $\begin{aligned}
                n\log(6+4\sqrt{2})-\Theta(\log n)
                \approx 3.54n-\Theta(\log n)
            \end{aligned}$ & & Theorem~\ref{thm:all-lb} \\
            \midrule
            \multirow[c]{2}{*}{$\Agen$} &
            $n\log 9-\Theta(\log n)\approx 3.17n-\Theta(\log n)$ & &
            \cite{jokim2025} \\
            & $\begin{aligned}
                n\log(7+4\sqrt{2})-\Theta(\log n)
                \approx 3.66n-\Theta(\log n)
            \end{aligned}$ & & Theorem~\ref{thm:all-lb} \\
            \bottomrule
        \end{tabular}}
    \caption{
        Summary of lower and upper bounds for the full query family.
        The stated lower bounds hold for both $\{\PSV,\PLV,\NSV,\NLV\}$ and $\{\RMin,\RMax\}$.
        The upper bounds support $\PSV$, $\PLV$, $\NSV$, $\NLV$, $\RMin$, and $\RMax$ simultaneously.
    }
    \label{tab:partial-results}
\end{table}

Table~\ref{tab:partial-results} compares our lower and upper bounds with previous results for the full query family.

More generally, we give a unified approach to deriving ITLBs for every subfamily of $\{\PSV,\PLV,\NSV,\NLV\}$ on the three array classes $\Perm$, $\Aneq$, and $\Agen$.
We also give succinct data structures for every such subfamily, matching the corresponding ITLBs up to lower-order terms while supporting queries in constant time.
Table~\ref{tab::result} (in Appendix~\ref{app:tables}) summarizes the lower and upper bounds for these query subfamilies across the three array classes.

The rest of the paper is organized as follows.
In the main text, we focus on the lower and upper bounds for the full query family.
Section~\ref{sec:pre} introduces the heap structures used throughout the paper and their connections to the query families.
Section~\ref{sec:lb} establishes the ITLBs for the full query family, and Section~\ref{sec:ub} presents matching succinct data structures supporting these queries in constant time.
Due to space constraints, some proofs from Sections~\ref{sec:lb} and~\ref{sec:ub} are deferred to Appendices~\ref{app:proof-lb} and~\ref{app:proof-ub}, respectively.
The lower bounds and data structures for the remaining query subfamilies are presented in Appendices~\ref{app:lb} and~\ref{app:ub}, respectively.

\section{Preliminaries}\label{sec:pre}

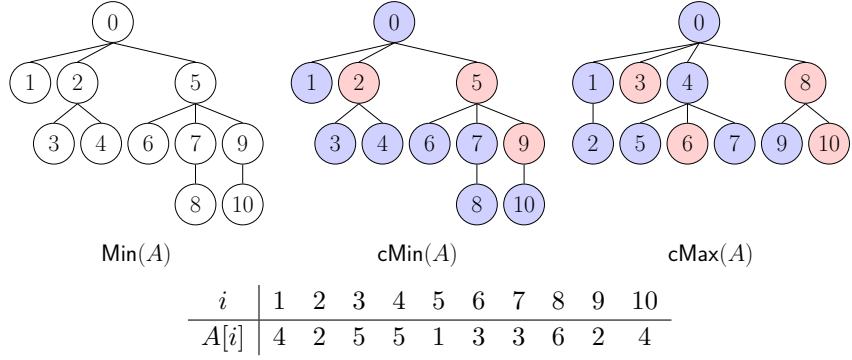
\begin{figure}[t]
\centering
\forestset{
    heap tree/.style={
        for tree={
            circle,
            draw,
            line width=1pt,
            minimum size=15mm,
            inner sep=0pt,
            font=\Huge,
            edge={line width=1pt},
            l sep=7mm,
            s sep=2mm,
            anchor=center,
            child anchor=north,
            parent anchor=south
        }
    },
    blue node/.style={fill=blue!18},
    red node/.style={fill=red!18}
}

\resizebox{0.8\textwidth}{!}{
\begin{tabular}{@{}c@{\hspace{10mm}}c@{\hspace{10mm}}c@{}}
\begin{forest}
baseline=(current bounding box.north),
heap tree
[0
    [1]
    [2
        [3]
        [4]
    ]
    [5
        [6]
        [7
            [8]
        ]
        [9
            [10]
        ]
    ]
]
\end{forest}
&
\begin{forest}
baseline=(current bounding box.north),
heap tree
[0, blue node
    [1, blue node]
    [2, red node
        [3, blue node]
        [4, blue node]
    ]
    [5, red node
        [6, blue node]
        [7, blue node
            [8, blue node]
        ]
        [9, red node
            [10, blue node]
        ]
    ]
]
\end{forest}
&
\begin{forest}
baseline=(current bounding box.north),
heap tree
[0, blue node
    [1, blue node
        [2, blue node]
    ]
    [3, red node]
    [4, blue node
        [5, blue node]
        [6, red node]
        [7, blue node]
    ]
    [8, red node
        [9, blue node]
        [10, red node]
    ]
]
\end{forest}
\\
\raisebox{-10mm}{{\Huge $\Min(A)$}} &
\raisebox{-10mm}{{\Huge $\cMin(A)$}} &
\raisebox{-10mm}{{\Huge $\cMax(A)$}}
\end{tabular}
}

\vspace{2mm}
{\centering
$
\begin{array}{c|cccccccccc}
i & 1 & 2 & 3 & 4 & 5 & 6 & 7 & 8 & 9 & 10\\
\hline
A[i] & 4 & 2 & 5 & 5 & 1 & 3 & 3 & 6 & 2 & 4\\
\end{array}
$
\par}
\vspace{-1mm}
\caption{Examples of $\Min(A)$, $\cMin(A)$, and $\cMax(A)$ for $A=4\;2\;5\;5\;1\;3\;3\;6\;2\;4$.}
\label{fig:heaps}
\end{figure}

\subparagraph{2d-min heaps.} Given an array $A[1,\dots,n]$, the \emph{2d-min heap} on $A$, denoted by $\Min(A)$, is a rooted ordered tree with nodes labeled $0,1,\dots,n$ satisfying the following properties: 
\begin{enumerate}
    \item the root is labeled $0$,
    \item for every $i>0$, the parent of the node labeled $i$ is the node labeled $\PSV(i)$, and
    \item the preorder traversal of $\Min(A)$ visits the nodes in the order of their labels, $0,1,\dots,n$.
\end{enumerate}
Similarly, the \emph{2d-max heap} on $A$, denoted by $\Max(A)$, is defined by replacing $\PSV$ with $\PLV$. Here, we simply denote a node labeled $x$ by a node $x$.

Fischer and Heun~\cite{fischerheun2011} showed that $\Min(A)$ supports $\PSV$ and $\RMin(\cdot, \cdot, n)$ queries on $A$. 
Thus, by maintaining both $\Min(A)$ and $\Max(A)$, $\PSV$, $\PLV$, $\RMin(\cdot, \cdot, n)$, and $\RMax(\cdot, \cdot, n)$ queries can be supported simultaneously.

\begin{algorithm}
\caption{Construction of the 2d-min heap.}\label{alg:minheap}
\DontPrintSemicolon
\KwIn{$A[1 \ldots n]$; assume $A[0]=-\infty$}
\KwOut{2d-min heap on $A$}

Initialize a stack $S$ and $S.\mathrm{push}(0)$\;
Set node $0$ as the root\;

\For{$i = 1 \ldots n$}{
    \While{$A[S.\mathrm{top}()] \ge A[i]$}{
        $S.\mathrm{pop}()$\;
    }
    $p := S.\mathrm{top}()$\;
    Append node $i$ as the rightmost child of node $p$\;
    $S.\mathrm{push}(i)$\;
}
\end{algorithm}

Algorithm~\ref{alg:minheap} gives the stack-based construction of $\Min(A)$ due to Fischer and Heun~\cite{fischerheun2011}.
The construction of $\Max(A)$ is analogous, with the comparison $\ge$ replaced by $\le$ and $A[0]=+\infty$ assumed.
We will use the following consequences of earlier work on the simultaneous construction of $\Min(A)$ and $\Max(A)$.

\begin{proposition}[\cite{gawrychowski2015, josatti2016, tsur2019}]\label{prop:minmax}
Consider the stacks maintained by the construction algorithms for $\Min(A)$ and $\Max(A)$. For each step $i>1$, the following holds:
\begin{itemize}
    \item If $A[i-1] \neq A[i]$, pops occur in exactly one of the two stacks. Consequently, in the heap corresponding to the stack with no pop, node $i$ becomes a leftmost child and node $i-1$ becomes an internal node; in the other heap, node $i$ becomes a non-leftmost child and node $i-1$ becomes a leaf.
    \item If $A[i-1] = A[i]$, exactly one pop occurs in each of the two stacks. Consequently, in both heaps, node $i$ is appended as the right sibling of node $i-1$ and node $i-1$ becomes a leaf.
\end{itemize}
\end{proposition}

\subparagraph{Colored 2d-min heaps.} The \emph{colored 2d-min heap} on $A$, denoted by $\cMin(A)$, is defined by coloring each node of $\Min(A)$ as follows:
\begin{enumerate}
    \item the root and leftmost children are colored blue, and
    \item for each node $x$ with left sibling $y$, node $x$ is colored blue if $A[x] = A[y]$ and red otherwise.
\end{enumerate}
Using the same coloring rule, the \emph{colored 2d-max heap} on $A$, denoted by $\cMax(A)$, is defined analogously.

Fischer~\cite{fischer2011} showed that $\cMin(A)$ additionally supports $\NSV$ on $A$. Jo and Satti~\cite{josatti2016} further showed that $\cMin(A)$ supports $\RMin$ queries. Thus, by maintaining both $\cMin(A)$ and $\cMax(A)$, $\PSV$, $\NSV$, $\PLV$, $\NLV$, $\RMin$, and $\RMax$ queries can be supported simultaneously.
We will use the following consequences for the colored heaps.

\begin{proposition}[\cite{tsur2019,jokim2025}]\label{prop:cminmax}
For each $i>1$, the following properties hold.
\begin{enumerate}
    \item If $A[i-1] \neq A[i]$, then, in either $\cMin(A)$ or $\cMax(A)$, whenever node $i$ has a leaf as its left sibling, the color of node $i$ in that heap is forced to be red.
    \item If $A[i-1] = A[i]$, then the color of node $i$ is forced to be blue in both heaps.
\end{enumerate}
\end{proposition}

Figure~\ref{fig:heaps} gives examples of the heaps and colored heaps defined above.

\section{Lower bounds}\label{sec:lb}

In this section, we establish ITLBs for supporting the full query family $\{\PSV,\PLV,\NSV,\NLV\}$. By the following proposition, these bounds also imply the corresponding lower bounds for the query family $\{\RMin,\RMax\}$.

\begin{proposition}\label{prop:query-eq}
    The following properties hold:
    \begin{enumerate}
        \item The query families $\{\PSV, \NSV\}$ and $\{\RMin\}$ can simulate each other.
        \item The query families $\{\PLV, \NLV\}$ and $\{\RMax\}$ can simulate each other.
    \end{enumerate}
\end{proposition}
\begin{proof}
    We prove the first statement; the second is analogous.

    We first simulate $\RMin$ using $\PSV$ and $\NSV$. It is known that $\PSV$ can simulate $\RMin(\cdot,\cdot,n)$, which returns the rightmost position among the minimum elements in a given interval~\cite{fischerheun2011}. Similarly, $\NSV$ can simulate $\RMin(\cdot,\cdot,1)$, which returns the leftmost position among the minimum elements in a given interval.

    Together, these two queries can simulate $\RMin(i,j,k)$ as follows. Let $p_1 := \RMin(i,j,1)$ and $r := \RMin(i,j,n)$. If $p_1=r$, then there is a single minimum position in $[i,j]$. Otherwise, starting from $p_1$, repeatedly set $p_{\ell+1} := \RMin(p_\ell+1,r,1)$ until $p_{\ell+1}=r$. This recovers all minimum positions in $[i,j]$ in increasing order, and therefore answers $\RMin(i,j,k)$.

    Conversely, we can also simulate $\PSV$ and $\NSV$ using $\RMin$. First consider $\PSV(i)$. If $i=1$, then $\PSV(i)=0$. Otherwise, for $t=i-1,i-2,\dots,1$, query $\RMin(t,i,n)$. The first position $t$ for which the answer is not $i$ is exactly $\PSV(i)$; if no such position exists, then $\PSV(i)=0$. The simulation of $\NSV(i)$ is symmetric.
\end{proof}

\subsection{Main ideas}

To compute desired ITLBs, we count the number of different inputs that give distinct answers to at least one query in $\{\PSV,\PLV,\NSV,\NLV\}$.

First, consider supporting only $\PSV$.
Two arrays have identical answers to all $\PSV$ queries if and only if they have the same 2d-min heap~\cite{fischerheun2011}.
Every ordinal tree with $n+1$ nodes can arise as the 2d-min heap of a permutation of length $n$, and hence of an array in any of the three classes.
These heaps are therefore counted by the Catalan numbers (\href{https://oeis.org/A000108}{A000108}), yielding an ITLB of $2n-\Theta(\log n)$ bits.

By this equivalence and its max-heap analogue, two arrays have identical answers to all $\PSV$ and $\PLV$ queries if and only if they induce the same pair $(\Min(A),\Max(A))$.
Thus, it suffices to count the distinct realizable pairs in the array class under consideration.
Here, not every pair of ordinal trees with $n+1$ nodes is \emph{realizable};
i.e., some pairs cannot arise as $(\Min(A),\Max(A))$ for any array $A$ in the class under consideration.
For permutations, the realizable pairs are counted by the Baxter numbers (\href{https://oeis.org/A001181}{A001181}), yielding an ITLB of $3n-\Theta(\log n)$ bits rather than $4n-\Theta(\log n)$ bits~\cite{dulucq1996stack,gawrychowski2015}.

Two arrays have identical answers to all $\PSV$ and $\NSV$ queries if and only if they have the same colored 2d-min heap~\cite{fischer2011}.
Together with the symmetric statement for $\PLV$ and $\NLV$, this implies that two arrays have identical answers to all queries in $\{\PSV,\PLV,\NSV,\NLV\}$ if and only if they induce the same pair $(\cMin(A),\cMax(A))$.
Hence, it suffices to count the realizable pairs of colored 2d-min and 2d-max heaps.
As in the $\{\PSV,\PLV\}$ case, counting these pairs requires accounting for dependencies between the heaps, now involving both their topologies and colors.
However, the earlier arguments~\cite{dulucq1996stack,gawrychowski2015} do not directly extend to this case; to our knowledge, no combinatorial class was previously known to capture the dependencies between the two colored heaps in the way that Baxter permutations do for the $\{\PSV,\PLV\}$ case.

To capture these dependencies systematically, we use a \emph{generating tree}.
Let $\mathcal{C}=\bigcup_{n\geq 1}\mathcal{C}_n$ be a combinatorial class, where $\mathcal{C}_n$ denotes the set of objects of size $n$, and suppose that $\mathcal{C}_1$ contains a unique object. A generating tree for $\mathcal{C}$ is a rooted tree satisfying the following properties:
\begin{itemize}
    \item the root corresponds to the unique object in $\mathcal{C}_1$;
    \item the children of a node corresponding to an object in $\mathcal{C}_n$ are determined by a prescribed rule that extends the object to objects in $\mathcal{C}_{n+1}$; and
    \item every object in $\mathcal{C}$ appears exactly once in the tree.
\end{itemize}

The stack-based construction in Algorithm~\ref{alg:minheap} scans the array from left to right, as does its max-heap counterpart.
Thus, appending a new element to the array extends each colored heap by one node, with the extension determined by the stack operations and the color of the new node.
We observe that Propositions~\ref{prop:minmax} and~\ref{prop:cminmax} clearly capture the dependencies between the two heaps, specifying which stack operations can occur simultaneously and which colors are forced.
Together, these constructions and properties allow us to characterize exactly how a realizable colored heap pair can be extended and thereby determine the children of each node in the generating tree.
We first rederive the generating tree for Baxter permutations given in~\cite{bousquet-melou2003} and then extend it to colored heap pairs.

From the generating trees, we derive functional equations for the associated generating functions.
For Baxter numbers, the obstinate kernel method yields an explicit formula~\cite{bousquet-melou2003}, but this method does not appear to be applicable to our equations for $\Aneq$ and $\Agen$.
Since asymptotic estimates suffice to derive the desired ITLBs, we follow the approach of Bouvel et al.~\cite{bouvel2018}, who relate the generating function for strong-Baxter permutations to that for quarter-plane excursions (defined in Section~\ref{subsec:gen-to-lb}).
Similarly, we relate our generating functions to those for quarter-plane excursions and apply known asymptotic results for these excursions.
Taking logarithms of the resulting counts yields the desired ITLBs.

This generating-tree approach provides a unified framework for deriving ITLBs for every subfamily of $\{\PSV,\PLV,\NSV,\NLV\}$ on the three array classes $\Perm$, $\Aneq$, and $\Agen$.
We apply this framework to the remaining query subfamilies in Appendix~\ref{app:lb}.

The remainder of this section is organized as follows.
Section~\ref{subsec:lb-min} illustrates the generating-tree approach using a single 2d-min heap.
Section~\ref{subsec:lb-minmax} extends this approach to pairs of uncolored 2d-min and 2d-max heaps.
Section~\ref{subsec:lb-cminmax} incorporates node colors to obtain generating trees for colored heap pairs.
Finally, Section~\ref{subsec:gen-to-lb} derives the ITLBs for the full query family on all three array classes from the resulting generating trees.

\subsection{Generating tree for 2d-min heaps}\label{subsec:lb-min}
In this subsection, we derive a generating tree for 2d-min heaps induced by arrays in any of the three classes $\Perm$, $\Aneq$, and $\Agen$.
Recall the construction algorithm for 2d-min heaps described in Algorithm~\ref{alg:minheap}.
The algorithm scans the array from left to right.
At each step, it pops indices from the stack until the top corresponds to a smaller value, attaches the new node as the rightmost child of that node, and pushes the new index onto the stack.

Consider the step of the construction algorithm that processes $A[m+1]$.
Let $s_0,s_1,\dots,s_k$ be the indices in the stack immediately before this step, listed from bottom to top, where $s_0=0$ and $s_k=m$.
Since $A[s_1]<\cdots<A[s_k]$, the number of indices popped depends on where $A[m+1]$ lies relative to these values.
Every pop count $r\in\{0,\dots,k\}$ can occur in each of the three array classes $\Perm$, $\Aneq$, and $\Agen$.

The nodes whose indices are stored in the stack form the rightmost spine of the current heap, namely the path from the root to node $m$.
We label each node of the generating tree by the stack size excluding index $0$, equivalently the number of edges on this spine.
The root of the generating tree corresponds to the unique heap induced by an array of length one, consisting of nodes $0$ and $1$, and has label $1$.
After $r$ pops and the push of index $m+1$, the label becomes $k-r+1$.
Since each pop count determines a distinct extension of the heap, a node labeled $k$ has $k+1$ children, labeled $1,2,\dots,k+1$, respectively.
Figure~\ref{fig:gen-tree} shows the first three levels of the generating tree for 2d-min heaps.

It is convenient to summarize a generating tree by its \emph{succession rule}, which specifies the label of the root and, for each possible label, the labels of the children of a node with that label. Here, the generating tree for 2d-min heaps is described by

\begin{equation}\label{rule:min}
    \text{Root: }(1) \qquad
    \text{Rule: }(k) \rightsquigarrow (1)(2)\cdots(k+1).
\end{equation}
This succession rule coincides with the succession rule for the Catalan numbers (\href{https://oeis.org/A000108}{A000108}).

\subsection{Combining 2d-min and 2d-max heaps}\label{subsec:lb-minmax}

In this subsection, we derive generating trees for pairs of 2d-min and 2d-max heaps induced by arrays in each of the three classes $\Perm$, $\Aneq$, and $\Agen$.
Assume we run the construction algorithms for the 2d-min and 2d-max heaps simultaneously, processing the same array element at each step.
We label each node of the generating tree by a pair $(k,\ell)$, where $k$ and $\ell$ are the stack sizes excluding index $0$ for the min-heap and max-heap constructions, respectively.

Consider the step that processes $A[m+1]$.
Let $(k,\ell)$ be the label immediately before this step.
We first consider the case where $A[m]\neq A[m+1]$.
By the first part of Proposition~\ref{prop:minmax}, pops occur in exactly one of the two stacks.
Suppose that no pop occurs in the stack of $\Min(A)$.
Its label then increases from $k$ to $k+1$.
In the stack of $\Max(A)$, any pop count from $1$ to $\ell$ can occur, so its label after pushing $m+1$ can be any value from $1$ to $\ell$.
The other case is symmetric.
Thus, for $\Perm$ and $\Aneq$, where consecutive elements are always distinct, we obtain the following succession rule:
\begin{equation}\label{rule:minmax-no-cons}
    \text{Root: }(1,1) \qquad
    (k,\ell) \rightsquigarrow
    \begin{cases}
        (1,\ell+1)(2,\ell+1)\cdots(k,\ell+1),\\
        (k+1,1)(k+1,2)\cdots(k+1,\ell).
    \end{cases}
\end{equation}
This succession rule coincides with the succession rule for the Baxter numbers (\href{https://oeis.org/A001181}{A001181}).
Figure~\ref{fig:gen-tree-pair} shows the first three levels of the generating tree for uncolored heap pairs.

We next consider the case where $A[m]=A[m+1]$.
By the second part of Proposition~\ref{prop:minmax}, exactly one pop occurs in each stack.
After pushing $m+1$, both labels remain unchanged.
Thus, for $\Agen$, we add a child node labeled $(k,\ell)$ to each node labeled $(k,\ell)$ in the generating tree defined by rule~\eqref{rule:minmax-no-cons}, obtaining
\begin{equation}\label{rule:minmax}
    \text{Root: }(1,1) \qquad
    \text{Rule: }(k,\ell) \rightsquigarrow
    \begin{cases}
        (1,\ell+1)(2,\ell+1)\cdots(k,\ell+1),\\
        (k+1,1)(k+1,2)\cdots(k+1,\ell),\\
        (k,\ell).
    \end{cases}
\end{equation}

\subsection{Combining colored 2d-min and 2d-max heaps}\label{subsec:lb-cminmax}

In this subsection, we derive generating trees for pairs of colored 2d-min and 2d-max heaps induced by arrays in each of the three classes $\Perm$, $\Aneq$, and $\Agen$.
Counting these pairs yields the ITLB for supporting the full query family on each array class.

For $\Perm$, the colors are determined by the heap structure: the root and all leftmost children are blue, while all other nodes are red.
Thus, the same succession rule~\eqref{rule:minmax-no-cons} applies to colored heap pairs on permutations.

We now turn to $\Aneq$ and $\Agen$.
Consider the step that processes $A[m+1]$, and let $(k,\ell)$ be the label immediately before this step.
We first consider the case where $A[m]\neq A[m+1]$.
By Proposition~\ref{prop:minmax}, the new node is a leftmost child in one of the two heaps and is therefore blue in that heap.
It remains to determine its color in the other heap.
If its left sibling is a leaf, then its color is forced to be red by Proposition~\ref{prop:cminmax}.
This happens exactly when one index is popped from the corresponding stack, giving the label $(k,\ell+1)$ or $(k+1,\ell)$.
Otherwise, the left sibling is internal, and its value may be equal to or different from $A[m+1]$, so either color can occur.
Thus, for $\Aneq$, we obtain the following succession rule, where a superscript $2$ denotes two children with the same label, one for each possible color:
\begin{equation}\label{rule:all-no-cons}
    \text{Root: }(1,1) \qquad
    (k,\ell)\rightsquigarrow
    \begin{cases}
        (1,\ell+1)^2(2,\ell+1)^2\cdots(k-1,\ell+1)^2(k,\ell+1),\\
        (k+1,1)^2(k+1,2)^2\cdots(k+1,\ell-1)^2(k+1,\ell).
    \end{cases}
\end{equation}
Figure~\ref{fig:gen-tree-colored-pair} shows the first three levels of this generating tree.
\begin{figure}[t]
\centering
\input{figures/gen_tree_pair_helpers.tex}
\def\GPMinName{\cMin}
\def\GPMaxName{\cMax}

\resizebox{0.95\linewidth}{!}{
\begin{tikzpicture}
    \GPObjectSide{root}{(0,0)}
        {[0,GPblue[1,GPblue]]}
        {[0,GPblue[1,GPblue]]}
        {1,1}

    \GPObjectSide{a}{(-5.8,-4.8)}
        {[0,GPblue[1,GPblue][2,GPred]]}
        {[0,GPblue[1,GPblue[2,GPblue]]]}
        {1,2}
    \GPObjectSide{b}{(5.8,-4.8)}
        {[0,GPblue[1,GPblue[2,GPblue]]]}
        {[0,GPblue[1,GPblue][2,GPred]]}
        {2,1}

    \GPObject{aa}{(-10.15,-12.3)}
        {[0,GPblue[1,GPblue][2,GPred][3,GPred]]}
        {[0,GPblue[1,GPblue[2,GPblue[3,GPblue]]]]}
        {1,3}
    \GPObject{ab}{(-7.25,-12.3)}
        {[0,GPblue[1,GPblue][2,GPred[3,GPblue]]]}
        {[0,GPblue[1,GPblue[2,GPblue]][3,GPblue]]}
        {2,1}
    \GPObject{ac}{(-4.35,-12.3)}
        {[0,GPblue[1,GPblue][2,GPred[3,GPblue]]]}
        {[0,GPblue[1,GPblue[2,GPblue]][3,GPred]]}
        {2,1}
    \GPObject{ad}{(-1.45,-12.3)}
        {[0,GPblue[1,GPblue][2,GPred[3,GPblue]]]}
        {[0,GPblue[1,GPblue[2,GPblue][3,GPred]]]}
        {2,2}

    \GPObject{ba}{(1.45,-12.3)}
        {[0,GPblue[1,GPblue[2,GPblue]][3,GPblue]]}
        {[0,GPblue[1,GPblue][2,GPred[3,GPblue]]]}
        {1,2}
    \GPObject{bb}{(4.35,-12.3)}
        {[0,GPblue[1,GPblue[2,GPblue]][3,GPred]]}
        {[0,GPblue[1,GPblue][2,GPred[3,GPblue]]]}
        {1,2}
    \GPObject{bc}{(7.25,-12.3)}
        {[0,GPblue[1,GPblue[2,GPblue][3,GPred]]]}
        {[0,GPblue[1,GPblue][2,GPred[3,GPblue]]]}
        {2,2}
    \GPObject{bd}{(10.15,-12.3)}
        {[0,GPblue[1,GPblue[2,GPblue[3,GPblue]]]]}
        {[0,GPblue[1,GPblue][2,GPred][3,GPred]]}
        {3,1}

    \draw[GPedge] (rootLabel.south) -- (a.north);
    \draw[GPedge] (rootLabel.south) -- (b.north);
    \foreach \child in {aa,ab,ac,ad} {
        \draw[GPedge] (aLabel.south) -- (\child.north);
    }
    \foreach \child in {ba,bb,bc,bd} {
        \draw[GPedge] (bLabel.south) -- (\child.north);
    }
\end{tikzpicture}
}
\caption{The first three levels of the generating tree for colored
2d-min and 2d-max heap pairs on arrays in $\Aneq$. The two heaps are
shown side by side in the first two levels and vertically in the third
level, with the label below. Children with the same label can differ in
their node colors.}
\label{fig:gen-tree-colored-pair}
\end{figure}

We next consider the case where $A[m]=A[m+1]$.
By Proposition~\ref{prop:minmax}, exactly one pop occurs in each stack, so both labels remain unchanged after pushing $m+1$.
Moreover, the new node is blue in both heaps by Proposition~\ref{prop:cminmax}.
Thus, for $\Agen$, we add a child node labeled $(k,\ell)$ to each node labeled $(k,\ell)$ in the generating tree defined by rule~\eqref{rule:all-no-cons}, obtaining
\begin{equation}\label{rule:all}
    \text{Root: }(1,1) \qquad
    (k,\ell)\rightsquigarrow
    \begin{cases}
        (1,\ell+1)^2(2,\ell+1)^2\cdots(k-1,\ell+1)^2(k,\ell+1),\\
        (k+1,1)^2(k+1,2)^2\cdots(k+1,\ell-1)^2(k+1,\ell),\\
        (k,\ell).
    \end{cases}
\end{equation}

\subsection{Lower bounds for supporting all queries}\label{subsec:gen-to-lb}

In this section, we give a proof of Theorem~\ref{thm:all-lb}, restated below.

\restAllLb*

As outlined above, we follow the approach of Bouvel et al.~\cite[Section~5.2]{bouvel2018}.
We derive functional equations for the generating functions associated with the succession rules and relate these generating functions to those of quarter-plane walks.
Known asymptotic results for such walks then yield the desired ITLBs.

\subparagraph{Quarter-plane walks.}

Given a finite step multiset $\mathcal S$ in $\mathbb Z^2$, a \textit{quarter-plane walk} with steps in $\mathcal S$ is a lattice path that starts at $(0,0)$, stays in $\mathbb Z_{\geq 0}^2$, and takes each step from $\mathcal S$. The \textit{length} of a walk is the number of its steps. We allow a step to have a multiplicity: two copies of the same step are counted as two different choices. Equivalently, multiple copies of the same step may be viewed as steps of different colors. For a step multiset $\mathcal S$, let
\[
 W_{\mathcal S}(t;a,b)
 =\sum_{n,i,j\geq0}w_{n,i,j}t^na^ib^j,
\]
where $w_{n,i,j}$ is the number of walks of length $n$ that end at $(i,j)$.
A walk ending again at $(0,0)$ is called an \textit{excursion}, so
$W_{\mathcal S}(t;0,0)$ is the generating function for excursions. The
\textit{characteristic polynomial} of $\mathcal S$ is defined by
\[
 Q_{\mathcal S}(a,b)=\sum_{(u,v)\in\mathcal S}a^ub^v.
\]

We use the following consequence of~\cite[Theorem~3.11]{bostan2021}, specialized to the step multisets considered here.
Its derivation is given in Appendix~\ref{app:proof-lb}.

\begin{restatable}{lemma}{restQuarterWalk}\label{lem:quarter-walk}
Let $\mathcal S$ be a finite step multiset in $\mathbb Z^2$ containing $(0,0)$, $(1,0)$, $(-1,0)$, $(0,1)$, and $(0,-1)$.
Then its characteristic polynomial $Q_{\mathcal S}$ has a unique critical point $(a,b)\in\mathbb R_{>0}^2$.
Writing $\mu=Q_{\mathcal S}(a,b)$, the number $e_n$ of excursions of length $n$ satisfies
\[
    e_n=\Theta\!\left(\frac{\mu^n}{n^c}\right)
\]
for some constant $c>0$ depending only on $\mathcal S$.
\end{restatable}

\subparagraph{Proof of Theorem~\ref{thm:all-lb}.}
We now derive the ITLBs for $\Aneq$ and $\Agen$ by relating the number of realizable colored heap pairs to the number of quarter-plane excursions.
This allows us to apply Lemma~\ref{lem:quarter-walk} to obtain the required asymptotic bounds.

Consider the step multiset
$
    \mathcal S'=\{(1,0),(0,1),2(-1,0),2(0,-1),2(1,-1),2(-1,1),2(0,0)\},
$
where $r(u,v)$ denotes $r$ copies of the step $(u,v)$.
Let $\mathcal S$ be obtained from $\mathcal S'$ by adding one further copy of $(0,0)$.

\begin{proposition}\label{prop:heap-walk}
    For every $n\ge1$, the numbers of distinct colored heap pairs $(\cMin(A),\cMax(A))$ induced by arrays in $\Aneq_n$ and $\Agen_n$ equal the numbers of excursions of length $n-1$ with steps in $\mathcal S'$ and $\mathcal S$, respectively.
\end{proposition}
\begin{proof}
    We prove the statement for $\Aneq$; the case of $\Agen$ is analogous.
    We begin by translating rule~\eqref{rule:all-no-cons} into a generating function.

    Let $f_{n,k,\ell}$ be the number of nodes at depth $n-1$ with label $(k,\ell)$ in the generating tree, and set
    \[
        \mathcal F(x;y,z)
        =\sum_{n\ge1}\sum_{k,\ell\ge1} f_{n,k,\ell}x^ny^kz^\ell .
    \]
    
    By rule~\eqref{rule:all-no-cons}, the root contributes $xyz$,
    and the children of a node labeled $(k,\ell)$ contribute the following polynomial in $y$ and $z$:
    \[
        2z^{\ell+1}\sum_{i=1}^{k-1}y^i+y^kz^{\ell+1}
        +2y^{k+1}\sum_{j=1}^{\ell-1}z^j+y^{k+1}z^\ell.
    \]
    Multiplying this polynomial by $x^{n+1}f_{n,k,\ell}$, summing over $n,k,\ell\ge1$, and adding the root contribution, we obtain the following equation after evaluating the geometric sums:
    \begin{align}\label{eq:all-gf}
    \mathcal F(x;y,z)
    &=xyz+\frac{2xz}{1-y}\bigl(y\mathcal F(x;1,z)-\mathcal F(x;y,z)\bigr)
        +xz\mathcal F(x;y,z) \notag\\
    &\quad+\frac{2xy}{1-z}\bigl(z\mathcal F(x;y,1)-\mathcal F(x;y,z)\bigr)
        +xy\mathcal F(x;y,z).
    \end{align}

    Define $\widehat{\mathcal F}(x;a,b)=\mathcal F(x;1+a,1+b)$.
    Substituting $y=1+a$ and $z=1+b$ into~\eqref{eq:all-gf} and rearranging terms, we obtain
    \begin{align}\label{eq:all-shifted}
        &(1-xQ_{\mathcal S'}(a,b))\widehat{\mathcal F}(x;a,b) \notag\\
        &\quad =x(1+a)(1+b) -\frac{2x(1+a)(1+b)}{a}\widehat{\mathcal F}(x;0,b) -\frac{2x(1+a)(1+b)}{b}\widehat{\mathcal F}(x;a,0),
    \end{align}
    where the characteristic polynomial of $\mathcal S'$ is
    \[
        Q_{\mathcal S'}(a,b)=a+b+\frac{2}{a}+\frac{2}{b}+\frac{2a}{b}+\frac{2b}{a}+2.
    \]

    Write $W(x;a,b)=W_{\mathcal S'}(x;a,b)$.
    Appending a step multiplies the generating function by $xQ_{\mathcal S'}(a,b)$, except that steps leaving the quarter plane must be excluded.
    The forbidden steps from the vertical boundary are the two copies each of $(-1,0)$ and $(-1,1)$, and those from the horizontal boundary are the two copies each of $(0,-1)$ and $(1,-1)$.
    Since no step decreases both coordinates, these two corrections do not overlap.
    Including the empty walk, we obtain
    \begin{equation}\label{eq:all-walk}
        (1-xQ_{\mathcal S'}(a,b))W(x;a,b)
        =1-\frac{2x(1+b)}{a}W(x;0,b)
        -\frac{2x(1+a)}{b}W(x;a,0).
    \end{equation}

    Define $G(x;a,b):=x(1+a)(1+b)W(x;a,b)$.
    Rewriting~\eqref{eq:all-walk} in terms of $G$ yields~\eqref{eq:all-shifted}, so $G=\widehat{\mathcal F}$.
    Setting $a=b=0$, we obtain
    \[
        \widehat{\mathcal F}(x;0,0)=\mathcal F(x;1,1)=xW(x;0,0).
    \]
    This proves the statement for $\Aneq$.
\end{proof}

We now complete the proof of Theorem~\ref{thm:all-lb}.
Both $\mathcal S'$ and $\mathcal S$ satisfy the assumptions of Lemma~\ref{lem:quarter-walk}.
Their characteristic polynomials have the same unique positive critical point $(\sqrt{2},\sqrt{2})$, at which their values are $6+4\sqrt{2}$ and $7+4\sqrt{2}$, respectively.
Applying Lemma~\ref{lem:quarter-walk} and Proposition~\ref{prop:heap-walk} and taking logarithms of the resulting counts gives the claimed ITLBs.
The corresponding bounds for $\{\RMin,\RMax\}$ follow from Proposition~\ref{prop:query-eq}.

\section{Upper bounds}\label{sec:ub}

We first introduce some notations used throughout this section.

\subparagraph{Word-access queries.}
Given a bit sequence $S$ and a position $i$, the \emph{word-access} query on $S$, denoted by $\wordacc(i)$, returns the length-$w$ substring of $S$ starting at position $i$, where $w=\Theta(\log n)$ is the word size.

Since we work in the word-RAM model with word size $w$, supporting this operation in constant time lets us access a represented bit sequence as if it were stored explicitly, with no asymptotic overhead in access time.

\subparagraph{Rank and select queries.} Given a bit array $B[1,\dots,n]$ and a pattern $p\in\{0,1\}^+$, the query $\rank_p(i,B)$ returns the number of occurrences of $p$ in $B[1,\dots,i]$, and $\select_p(j,B)$ returns the first position of the $j$-th occurrence of $p$ in $B$. We will use the following lemma.

\begin{lemma}[\cite{rrr2007, munro2001}]\label{lem:rrr}
    Let $B[1,\dots,n]$ be a bit array with $m$ ones. Then $B$ can be stored in $\log\binom{n}{m}+o(n)\le n+o(n)$ bits while supporting $\wordacc$, $\rank_p$ and $\select_p$ in constant time where $p\in\{0, 1\}$.
\end{lemma}

\subparagraph{Balanced parentheses.}
For an ordinal tree $T$, the \emph{balanced-parentheses sequence} of $T$, denoted by $\BP(T)$, is obtained by a preorder traversal of $T$: append an opening parenthesis when a node is first visited, and append a closing parenthesis when the traversal leaves the node.
Thus, if $T$ has $n$ nodes, $\BP(T)$ has length $2n$ and can be stored explicitly using $2n$ bits.
With $o(n)$ additional bits, $\BP(T)$ supports various tree navigation queries on $T$ in constant time~\cite{navarro2014}.

In this section, we mainly use balanced-parentheses sequences to represent 2d-min and 2d-max heaps.
Consider $\Min(A)$ for an array $A$.
We use the following observation of \cite{jokim2025}: the number of closing parentheses immediately preceding the opening parenthesis of node $i$ in $\BP(\Min(A))$ is exactly the number of stack pops performed just before node $i$ is appended in the construction algorithm.

\subparagraph{Color sequences.}
Given an array $A$, let $\Xmin(A)$ be the color bit sequence of $\cMin(A)$, storing the colors of the nodes in preorder, where $0$ denotes blue and $1$ denotes red. Define $\Xmax(A)$ analogously for $\cMax(A)$.

\subsection{Main ideas and key lemmas}

We first outline our approach and state the key lemmas used in our constructions.
Their proofs are deferred to Appendix~\ref{app:proof-ub}.

Our approach builds on the combined representation of $(\Min(A),\Max(A))$ for arrays $A\in\Aneq_n$ proposed by Jo and Kim~\cite{jokim2025}.
They exploit the dependency between the two heaps described in Proposition~\ref{prop:minmax} to reduce the space needed to represent their BP sequences.
Specifically, their representation consists of two bit arrays $U$ and $S$, defined in Section~\ref{subsec:ub-all}, of lengths $n-1$ and at most $2n-2$ bits, respectively.

To represent the colored heap pair $(\cMin(A),\cMax(A))$, they additionally store color information, recording only the colors that are not determined by the coloring rule and Proposition~\ref{prop:cminmax}.
Their encoding of $S$, together with the color information, uses $(1+\log 3)n+o(n)$ bits, yielding a total space usage of $(2+\log 3)n+o(n)$ bits.
With $o(n)$ additional bits, their representation supports $\wordacc$ on both $\BP(\Min(A))$ and $\BP(\Max(A))$ in constant time.
Based on these word-access operations, they support the desired queries in $O(\log^{(\ell)}n)$ time for any positive constant integer $\ell$.

We improve both the space bound and the query time.
To improve the space bound, we first show that $S$ is a contiguous substring of $\BP(T)$ for some ordinal tree $T$; see Proposition~\ref{prop:excess}.
We further show that the stored color bits can be assigned to the non-leftmost children of $T$, yielding a colored 2d-min heap realizable by an array in $\Agen$.
By storing $U$ and applying Fischer's succinct representation~\cite{fischer2011} to this colored heap, we can achieve the desired space bound $n+n\log(3+2\sqrt{2})+o(n)=n\log(6+4\sqrt{2})+o(n)$.

However, Fischer's representation uses tree covering~\cite{farzan2008} and does not directly provide constant-time $\wordacc$ on $\BP(T)$.
We need this operation to support $\wordacc$ on $S$ and, in turn, on both $\BP(\Min(A))$ and $\BP(\Max(A))$.
We therefore develop an alternative succinct representation of colored 2d-min heaps that stores the tree topology and the color information separately.
The following lemma summarizes this representation.

\begin{restatable}{lemma}{restUbCmin}\label{lem:ub-cmin}
    Given an array $A\in\Agen_n$, there is a data structure using $n\log(3+2\sqrt{2})+o(n)\approx 2.54n+o(n)$ bits that supports $\wordacc$ on both $\BP(\Min(A))$ and $\Xmin(A)$ in constant time.
\end{restatable}

To improve the query time, we adapt the tree-covering representation of Hamada et al.~\cite{hamada2024}, which uses balanced parentheses.
Their representation allows the BP sequence of each micro tree to be obtained in constant time from the global BP sequence.
Similarly, given constant-time $\wordacc$ on $\BP(\Min(A))$ and $\Xmin(A)$, we obtain the BP and color sequences of each micro tree of $\cMin(A)$ in constant time.
We use each pair of local sequences as an index into a precomputed table to support local operations in constant time.
The following lemma formalizes this query support.

\begin{restatable}{lemma}{restCbp}\label{lem:cbp}
    Given an array $A\in\Agen_n$, suppose there is a data structure that supports $\wordacc$ on both $\BP(\Min(A))$ and $\Xmin(A)$ in constant time.
    Then, with $o(n)$ additional bits, $\PSV$, $\NSV$, and $\RMin$ on $A$ can be supported in constant time.
\end{restatable}

Finally, we show that our representation supports constant-time $\wordacc$ on $\BP(\Min(A))$, $\BP(\Max(A))$, $\Xmin(A)$, and $\Xmax(A)$.
Applying Lemma~\ref{lem:cbp} to both colored heaps then yields constant-time support for the desired queries.
This proves Theorem~\ref{thm:all-ub} for arrays in $\Aneq_n$.
We then extend our construction to general arrays by applying the reduction of Jo and Kim~\cite{jokim2025}.

As a further consequence, combining Lemmas~\ref{lem:ub-cmin} and~\ref{lem:cbp} yields the following corollary.

\begin{corollary}\label{cor:ub-cmin}
    Given an array $A\in\Agen_n$, there is a succinct data structure using $n\log(3+2\sqrt{2})+o(n) \approx 2.54n+o(n)$ bits that supports $\PSV$, $\NSV$, and $\RMin$ on $A$ in constant time.
\end{corollary}

\subsection{Encodings for supporting all queries}\label{subsec:ub-all}

In this section, we give a proof of Theorem~\ref{thm:all-ub}, restated below.

\restAllUb*

We first describe the representation of Jo and Kim~\cite{jokim2025} and then show how to improve it.
For an array $A\in\Aneq_n$, their representation consists of the following components:
\begin{itemize}
    \item $U[1,\dots,n-1]$: $U[i]=0$ iff node $i$ is an internal node in $\cMin(A)$.
    By Proposition~\ref{prop:minmax}, equivalently, $U[i]=1$ iff node $i$ is an internal node in $\cMax(A)$.
    \item $S=s_2\cdots s_n$: For $i\ge 2$, let $H\in\{\Min(A), \Max(A)\}$ be the heap in which node $i$ is a non-leftmost child of its parent.
    Let $p_i$ be the number of closing parentheses immediately preceding the opening parenthesis of node $i$ in $\BP(H)$.
    Equivalently, $p_i$ is the number of stack pops performed just before node $i$ is appended in the construction of $H$.
    We define $s_i=1^{p_i-1}0$.
    By Proposition~\ref{prop:minmax}, no pop occurs in the other construction, so $S$ encodes all close-runs for both $\Min(A)$ and $\Max(A)$.
    \item $c_{\min}$ (resp. $c_{\max}$): the sequence of colors of nodes in $\cMin(A)$ (resp. $\cMax(A)$) that are non-leftmost children and whose left sibling is not a leaf, listed in increasing order of their preorder numbers.
    By Propositions~\ref{prop:minmax} and~\ref{prop:cminmax}, all other colors are forced.
\end{itemize}

\begin{figure}[t]
\centering
\forestset{
    heap tree/.style={
        for tree={
            circle,
            draw,
            line width=1pt,
            minimum size=13mm,
            inner sep=0pt,
            font=\Huge,
            edge={line width=1pt},
            l sep=6mm,
            s sep=2mm,
            anchor=center,
            child anchor=north,
            parent anchor=south
        }
    },
    blue node/.style={fill=blue!18},
    red node/.style={fill=red!18}
}

\resizebox{0.45\textwidth}{!}{%
\begin{tabular}{@{}c@{\hspace{18mm}}c@{}}
\begin{forest}
baseline=(current bounding box.north),
heap tree
[0, blue node
    [1, blue node]
    [2, red node
        [3, blue node]
    ]
    [4, red node
        [5, blue node
            [6, blue node]
        ]
        [7, red node
            [8, blue node]
        ]
    ]
]
\end{forest}
&
\begin{forest}
baseline=(current bounding box.north),
heap tree
[0, blue node
    [1, blue node
        [2, blue node]
    ]
    [3, red node
        [4, blue node]
        [5, red node]
    ]
    [6, red node
        [7, blue node]
        [8, red node]
    ]
]
\end{forest}
\\
\raisebox{-7mm}{{\Huge $\cMin(A)$}} &
\raisebox{-7mm}{{\Huge $\cMax(A)$}}
\end{tabular}%
}

\vspace{2mm}
{\small
\setlength{\arraycolsep}{4pt}
$
\begin{array}{|c!{\color{black!25}\vrule}c!{\color{black!25}\vrule}c!{\color{black!25}\vrule}c!{\color{black!25}\vrule}c!{\color{black!25}\vrule}c!{\color{black!25}\vrule}c!{\color{black!25}\vrule}c!{\color{black!25}\vrule}c!{\color{black!25}\vrule}c!{\color{black!25}\vrule}c|}
\hline
i & 0 & 1 & 2 & 3 & 4 & 5 & 6 & 7 & 8 & \\
\hline
A[i] &  & 4 & 2 & 5 & 1 & 3 & 6 & 2 & 4 &  \\
\BP(\cMin(A)) & ( & ( & )( & ( & ))( & ( & ( & ))( & ( & )))) \\
\BP(\cMax(A)) & ( & ( & ( & ))( & ( & )( & ))( & ( & )( & ))) \\
U &  & 1 & 0 & 1 & 0 & 0 & 1 & 0 &  &  \\
S=s_2\cdots s_n &  &  & 0 & 10 & 10 & 0 & 10 & 10 & 0 &  \\
c_{\min} &  &  &  &  & 1 &  &  & 1 &  &  \\
c_{\max} &  &  &  & 1 &  &  & 1 &  &  &  \\
\hline
\end{array}
$
}
\vspace{-1mm}
\caption{An example of the sequences used in \Cref{subsec:ub-all} for an array $A\in\Aneq_8$.}
\label{fig:cminmax-no-cons}
\end{figure}

See Figure~\ref{fig:cminmax-no-cons} for an example of the above sequences.
We now describe our representation of $S$.
In what follows, interpret $0$ as an open parenthesis and $1$ as a close parenthesis in $S$.
For a bit array $B$ and a position $j$, define $\excess_B(j)$ as the number of zeros minus the number of ones in $B[1,\dots,j]$.
We also define $\excess(B)=\excess_B(|B|)$, with $\excess(B)=0$ if $B$ is empty.
We prove the following proposition.

\begin{proposition}\label{prop:excess}
    $S$ has nonnegative excess at every position.
\end{proposition}
\begin{proof}
    We construct $S$ by appending $s_i$ for $i=2,\dots,n$, in parallel with the two stack constructions of $\Min(A)$ and $\Max(A)$. Let $S'$ denote the current prefix constructed so far, and let $h$ be the sum of the two stack sizes, excluding the bottom index $0$ from each stack.

    We first maintain the invariant $\excess(S')=h-2$ after each string $s_i$ is appended. Initially, before $s_2$ is appended, $S'$ is empty, and each stack contains only the indices $0$ and $1$; hence $h=2$ and $\excess(S')=0=h-2$. Now consider step $i\ge2$. By Proposition~\ref{prop:minmax}, pops occur in exactly one of the two stacks. If $p_i$ indices are popped from that stack, then $h$ changes by $2-p_i$, since one stack loses $p_i$ indices and both stacks push index $i$. On the other hand, appending $s_i=1^{p_i-1}0$ to $S'$ changes $\excess(S')$ by $-(p_i-1)+1=2-p_i$. Thus the invariant is preserved.

    It remains to check that the excess does not become negative while each string $s_i$ is being appended. Before appending $s_i$, the stack in which the pops occur contains the $p_i$ indices to be popped, and the other stack contains at least one non-bottom index. Thus $h\ge p_i+1$, so $\excess(S')=h-2\ge p_i-1$ by the invariant. During the first $p_i-1$ symbols of $s_i$, the excess decreases by exactly $p_i-1$, and therefore never becomes negative. The final symbol is an open parenthesis, so it only increases the excess. Hence $S$ has nonnegative excess at every position.
\end{proof}

By the proposition, prepending an opening parenthesis to $S$ and appending $\excess(S)+1$ closing parentheses yields $\BP(T)$ for some ordinal tree $T$ with $n$ nodes.
Let $k=|c_{\min}|+|c_{\max}|$.

We now show that $T$ has $k$ non-leftmost child nodes.
In the tree $T$, the $i$-th node in preorder is a non-leftmost child exactly when $p_i\ge2$, since $s_i=1^{p_i-1}0$ and hence the opening parenthesis corresponding to this node is preceded by a closing parenthesis exactly in this case.
In terms of $\cMin(A)$ and $\cMax(A)$, the condition $p_i\ge2$ means that, in the heap where node $i$ is inserted as a non-leftmost child, its left sibling is internal.
By the definition of $c_{\min}$ and $c_{\max}$, such indices $i$ are in one-to-one correspondence with the color bits stored in these two sequences.
Hence $T$ has $k$ non-leftmost child nodes.

Let $Y=c_{\min}\cdot c_{\max}$.
We assign the bits of $Y$, in order, to the non-leftmost children of $T$ listed in preorder, and color the root and all leftmost children blue.
The resulting colored tree is realizable as a colored 2d-min heap of an array in $\Agen$.
We also store $|c_{\min}|$ using $O(\log n)$ bits to locate the boundary between $c_{\min}$ and $c_{\max}$.
Thus, by using Lemma~\ref{lem:ub-cmin}, together with $U$ and all auxiliary structures, the total space usage is $n+n\log(3+2\sqrt{2})+o(n)=n\log(6+4\sqrt{2})+o(n)$ bits.

We now extend the construction to arrays in $\Agen_n$.
Given an array $A\in\Agen_n$, define a bit array $C[1,\dots,n]$ by setting $C[1]=0$ and, for $i\ge2$, setting $C[i]=1$ iff $A[i-1]=A[i]$.
Let $A'\in\Aneq_m$ be the subsequence of $A$ induced by the positions $i$ with $C[i]=0$, where $m$ is the number of zeros in $C$.
Our representation for $A$ consists of the above representation for $A'$, together with $C$.
See Figure~\ref{fig:cminmax} for an example of this representation.

By Lemma~\ref{lem:rrr}, $C$ can be stored using $\log\binom{n}{m}+o(n)$ bits.
For any constant $c$ and integers $0\le m\le n$, we have $cm+\log\binom{n}{m}\le n\log(1+2^c)$~\cite{tsur2019}.
Applying this inequality with $c=\log(6+4\sqrt{2})$, the total space usage is at most $m\log(6+4\sqrt{2})+\log\binom{n}{m}+o(n)\le n\log(7+4\sqrt{2})+o(n)$ bits.

The query support follows from the following lemma (proof in Appendix~\ref{app:proof-ub}).

\begin{restatable}{lemma}{restAllWordacc}\label{lem:all-wordacc}
    For arrays $A$ in either $\Aneq_n$ or $\Agen_n$, the corresponding representation constructed in this subsection can be augmented with $o(n)$ bits to support constant-time $\wordacc$ on $\BP(\Min(A))$, $\BP(\Max(A))$, $\Xmin(A)$, and $\Xmax(A)$.
\end{restatable}

Applying Lemma~\ref{lem:cbp} to both colored heaps then yields constant-time support for the desired queries, completing the proof of Theorem~\ref{thm:all-ub}.

\section{Conclusions}
We gave a unified approach to the encoding complexity of every subfamily of $\{\PSV,\PLV,\NSV,\NLV\}$ queries on permutations, arrays without consecutive equal elements, and general arrays. For all the possible subfamilies, we gave both lower bounds on encoding space and matching upper bounds while supporting constant time queries.

Although we settle the succinctness of all these query subfamilies, there are some interesting open problems to extend these results. First, our redundancy term ($O(n\log\log n /\log n)$ bits) is not well-optimized. As in the encoding of range minimum queries~\cite{navarro2014}, reducing the redundancies of the data structure would be interesting. 
Also, there exist redundancy-query time trade-offs for range minimum queries~\cite{liu2020lower}. Extending the results for the queries considered in this paper would also be an interesting direction for future work.

\subparagraph{AI Usage Disclosure}
The authors used OpenAI's ChatGPT (GPT-5 through GPT-5.5) as a research aid through extensive, iterative discussions over several months.
GPT-5.6 Sol was additionally used to refine the language, clarity, and presentation of the manuscript.

The AI's main contribution concerned the lower-bound analysis.
It first suggested applying the generating-tree framework of Bousquet-Mélou~\cite{bousquet-melou2003} to the incremental construction in Algorithm~\ref{alg:minheap}, leading to the generating trees and associated functional equations used in this paper.
After extensive discussions about how to analyze these equations, the AI suggested relating them to quarter-plane walks based on the approaches of Bouvel et al.~\cite{bouvel2018} and Bostan et al.~\cite{bostan2021}.
This connection led to the asymptotic estimates needed to establish the lower bounds.
All claims and proofs were independently verified by the authors.

\bibliography{ref}

\clearpage

\appendix

\begingroup
\raggedbottom
\setlength{\intextsep}{1pt}

\section{Omitted Figures and Tables}\label{app:tables}

\begin{table}[H]
\centering
\resizebox{\textwidth}{!}{
\begin{tblr}{
    colspec={|c|c|r|c|c|r|c|c|},
    cells={valign=m,halign=c},
    column{3,6}={halign=r},
}
\hline
\SetCell[c=2]{c} Data structure &
&
\SetCell[c=3]{c} Lower bound &
&
&
\SetCell[c=3]{c} Upper bound &
&
\\
\cline{1-8}
Queries &
Array type &
\SetCell{c} Space (in bits) &
Sequence &
Reference &
\SetCell{c} Space (in bits) &
Query time &
Reference \\
\hline
{$\PSV, \RMin(\cdot,\cdot,n)$} &
$\Perm, \Aneq, \Agen$ &
$2n-\Theta(\log n)$ &
{Catalan\\\href{https://oeis.org/A000108}{A000108}} &
\cite{fischerheun2011} &
$2n+o(n)$ &
$O(1)$ &
\cite{fischerheun2011}
\\
\hline
\SetCell[r=3]{c} {$\PSV,\NSV$\\$\RMin$} &
$\Perm$ &
$2n-\Theta(\log n)$ &
{Catalan\\\href{https://oeis.org/A000108}{A000108}} &
\cite{fischerheun2011} &
$2n+o(n)$ &
$O(1)$ &
\cite{fischerheun2011, fischer2011}
\\
\cline{2-8}
&
$\Aneq$ &
$
\begin{aligned}
n\log(2+2\sqrt{2})-\Theta(\log n)\\
\approx 2.27n-\Theta(\log n)
\end{aligned}
$ &
\href{https://oeis.org/A071356}{A071356} &
Theorem~\ref{thm:cmin-no-cons-lb} &
$
\begin{aligned}
n\log(2+2\sqrt{2})+o(n)\\
\approx 2.27n+o(n)
\end{aligned}
$ &
$O(1)$ &
Theorem~\ref{thm:ub-cmin-no-cons}
\\
\cline{2-8}
&
$\Agen$ &
$
\begin{aligned}
n\log(3+2\sqrt{2})-\Theta(\log n)\\
\approx 2.54n-\Theta(\log n)
\end{aligned}
$ &
{Little Schr\"oder\\\href{https://oeis.org/A001003}{A001003}} &
\cite{fischer2011} &
$
\begin{aligned}
n\log(3+2\sqrt{2})+o(n)\\
\approx 2.54n+o(n)
\end{aligned}
$ &
$O(1)$ &
{\cite{fischer2011, josatti2016},\\
Corollary~\ref{cor:ub-cmin}}
\\
\hline
\SetCell[r=2]{c} {$\PSV,\PLV$\\$\RMin(\cdot,\cdot,n),$\\$\RMax(\cdot, \cdot, n)$} &
$\Perm, \Aneq$ &
$3n-\Theta(\log n)$ &
{Baxter\\\href{https://oeis.org/A001181}{A001181}} &
\cite{dulucq1996stack, gawrychowski2015} &
$
\begin{aligned}
3n+o(n)
\end{aligned}
$ &
$O(1)$ &
\cite{gawrychowski2015,josatti2016}
\\
\cline{2-8}
&
$\Agen$ &
$
\begin{aligned}
n\log 9-\Theta(\log n)\\
\approx 3.17n-\Theta(\log n)
\end{aligned}
$ &
\href{https://oeis.org/A216947}{A216947} &
Theorem~\ref{thm:lb-minmax} &
$
\begin{aligned}
n\log 9+o(n)\\
\approx 3.17n+o(n)
\end{aligned}
$ &
$O(1)$ &
Theorem~\ref{thm:ub-minmax}
\\
\hline
\SetCell[r=1]{c} {$\PSV,\NLV$\\$\RMin(\cdot,\cdot,n),$\\$\RMax(\cdot, \cdot, 1)$} &
$\Perm, \Aneq, \Agen$ &
$3n-\Theta(\log n)$ &
{Baxter\\\href{https://oeis.org/A001181}{A001181}} &
Theorem~\ref{thm:minrmax-lb} &
$
\begin{aligned}
3n+o(n)
\end{aligned}
$ &
$O(1)$ &
Theorem~\ref{thm:ub-minrmax}
\\
\hline
\SetCell[r=3]{c} {$\PSV,\PLV,\NSV,$\\$\RMin,\RMax(\cdot, \cdot, n)$} &
$\Perm$ &
$3n-\Theta(\log n)$ &
{Baxter\\\href{https://oeis.org/A001181}{A001181}} &
\cite{dulucq1996stack, gawrychowski2015} &
$
\begin{aligned}
3n+o(n)
\end{aligned}
$ &
$O(1)$ &
\cite{fischer2011, gawrychowski2015, josatti2016}
\\
\cline{2-8}
&
$\Aneq$ &
$
\begin{aligned}
n\log(4+4\sqrt{2})-\Theta(\log n)\\
\approx 3.27n-\Theta(\log n)
\end{aligned}
$ &
{Not indexed\\by the OEIS} &
Theorem~\ref{thm:three-lb} &
$
\begin{aligned}
n\log(4+4\sqrt{2})+o(n)\\
\approx 3.27n+o(n)
\end{aligned}
$ &
$O(1)$ &
Theorem~\ref{thm:ub-three}
\\
\cline{2-8}
&
$\Agen$ &
$
\begin{aligned}
n\log(5+4\sqrt{2})-\Theta(\log n)\\
\approx 3.41n-\Theta(\log n)
\end{aligned}
$ &
{Not indexed\\by the OEIS}
&
Theorem~\ref{thm:three-lb} &
$
\begin{aligned}
n\log(5+4\sqrt{2})+o(n)\\
\approx 3.41n+o(n)
\end{aligned}
$ &
$O(1)$ &
Theorem~\ref{thm:ub-three}
\\
\hline
\SetCell[r=3]{c} {$\PSV,\PLV,$\\$\NSV,\NLV,$\\$\RMin,\RMax$} &
$\Perm$ &
$3n-\Theta(\log n)$ &
{Baxter\\\href{https://oeis.org/A001181}{A001181}} &
\cite{dulucq1996stack, gawrychowski2015} &
$
\begin{aligned}
3n+o(n)
\end{aligned}
$ &
$O(1)$ &
\cite{fischer2011, gawrychowski2015, josatti2016}
\\
\cline{2-8}
&
$\Aneq$ &
{$
\begin{aligned}
n\log(6+4\sqrt{2})-\Theta(\log n)\\
\approx 3.54n-\Theta(\log n)
\end{aligned}
$} &
{Not indexed\\by the OEIS} &
Theorem~\ref{thm:all-lb} &
$
\begin{aligned}
n\log(6+4\sqrt{2})+o(n)\\
\approx 3.54n+o(n)
\end{aligned}
$ &
$O(1)$ &
Theorem~\ref{thm:all-ub}
\\
\cline{2-8}
&
$\Agen$ &
{$
\begin{aligned}
n\log(7+4\sqrt{2})-\Theta(\log n)\\
\approx 3.66n-\Theta(\log n)
\end{aligned}
$} &
{Not indexed\\by the OEIS} &
Theorem~\ref{thm:all-lb} &
$
\begin{aligned}
n\log(7+4\sqrt{2})+o(n)\\
\approx 3.66n+o(n)
\end{aligned}
$ &
$O(1)$ &
Theorem~\ref{thm:all-ub}
\\
\hline
\end{tblr}
}
\caption{
    Summary of lower and upper bounds.
    In each entry, the previous/next smaller/larger value queries and the range queries define two query families with the same ITLB.
    The upper bounds are for supporting all listed queries simultaneously.
}
\label{tab::result}
\end{table}

\begin{figure}[H]
\centering
\input{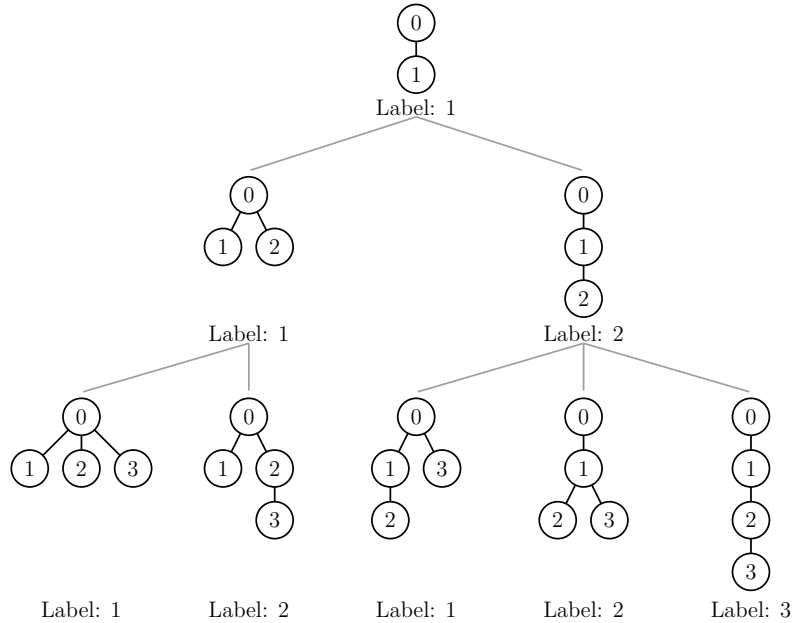}

\def\GenTreeRootToLevelOne{3.5}
\def\GenTreeLevelOneToLevelTwo{4.5}
\def\GenTreeEdgeToRootGap{3pt}
\def\GenTreeRootLabelOffset{2.1}
\def\GenTreeLevelOneLabelOffset{3.2}
\def\GenTreeLevelTwoLabelOffset{4.3}
\newcommand{\DrawGenTreeEdge}[2]{
    \draw[GPedge] (#1.south) -- ($(#2.north)+(0,\GenTreeEdgeToRootGap)$);
}

\resizebox{0.75\linewidth}{!}{
\begin{tikzpicture}
    \GPSingleObject{root}{(0,0)}
        {[0[1]]}{1}{\GenTreeRootLabelOffset}

    \GPSingleObject{a}{(-3.4,-\GenTreeRootToLevelOne)}
        {[0[1][2]]}{1}{\GenTreeLevelOneLabelOffset}
    \GPSingleObject{b}{(3.4,-\GenTreeRootToLevelOne)}
        {[0[1[2]]]}{2}{\GenTreeLevelOneLabelOffset}

    \GPSingleObject{aa}{(-6.8,{-\GenTreeRootToLevelOne-\GenTreeLevelOneToLevelTwo})}
        {[0[1][2][3]]}{1}{\GenTreeLevelTwoLabelOffset}
    \GPSingleObject{ab}{(-3.4,{-\GenTreeRootToLevelOne-\GenTreeLevelOneToLevelTwo})}
        {[0[1][2[3]]]}{2}{\GenTreeLevelTwoLabelOffset}
    \GPSingleObject{ba}{(0,{-\GenTreeRootToLevelOne-\GenTreeLevelOneToLevelTwo})}
        {[0[1[2]][3]]}{1}{\GenTreeLevelTwoLabelOffset}
    \GPSingleObject{bb}{(3.4,{-\GenTreeRootToLevelOne-\GenTreeLevelOneToLevelTwo})}
        {[0[1[2][3]]]}{2}{\GenTreeLevelTwoLabelOffset}
    \GPSingleObject{bc}{(6.8,{-\GenTreeRootToLevelOne-\GenTreeLevelOneToLevelTwo})}
        {[0[1[2[3]]]]}{3}{\GenTreeLevelTwoLabelOffset}

    \DrawGenTreeEdge{rootLabel}{a}
    \DrawGenTreeEdge{rootLabel}{b}
    \DrawGenTreeEdge{aLabel}{aa}
    \DrawGenTreeEdge{aLabel}{ab}
    \DrawGenTreeEdge{bLabel}{ba}
    \DrawGenTreeEdge{bLabel}{bb}
    \DrawGenTreeEdge{bLabel}{bc}
\end{tikzpicture}
}
\caption{The first three levels of the generating tree for 2d-min heaps. Each heap is shown together with its label.}
\label{fig:gen-tree}
\end{figure}
\begin{figure}[H]
\centering
\input{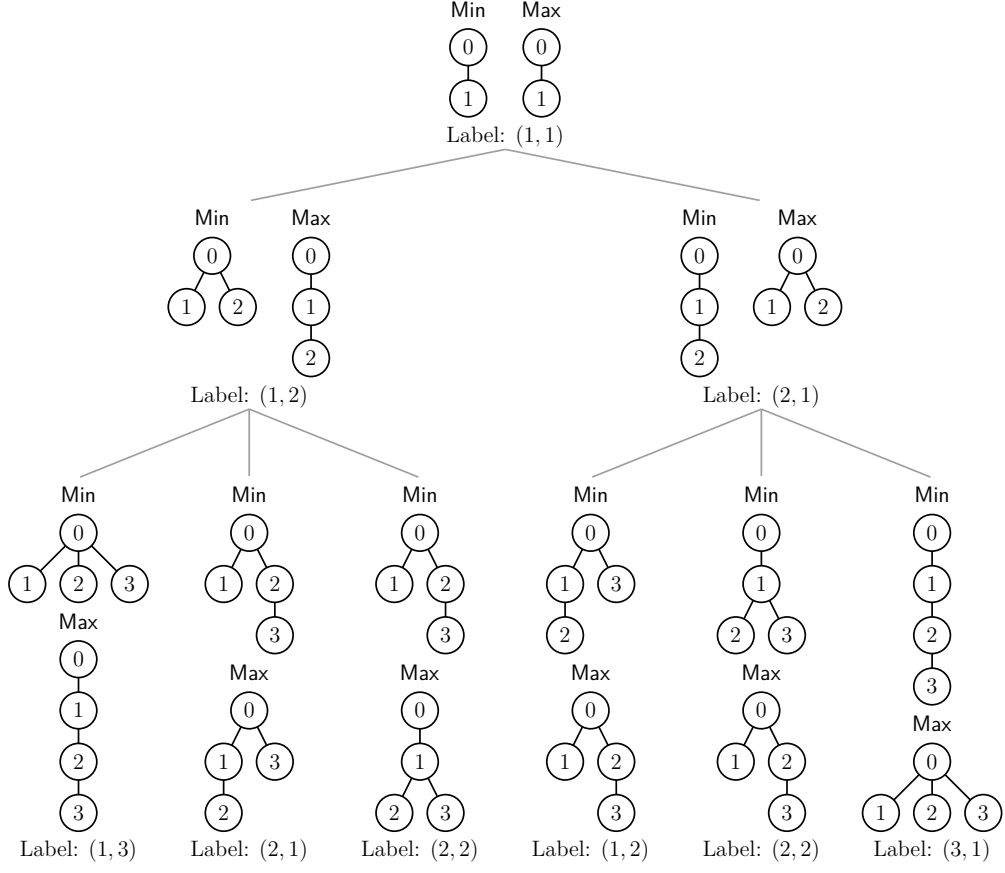}
\def\GPMinName{\Min}
\def\GPMaxName{\Max}
\resizebox{0.95\linewidth}{!}{
\begin{tikzpicture}
    \GPObjectSide{root}{(0,0)}
        {[0[1]]}{[0[1]]}{1,1}

    \GPObjectSide{a}{(-5.25,-4.8)}
        {[0[1][2]]}{[0[1[2]]]}{1,2}
    \GPObjectSide{b}{(5.25,-4.8)}
        {[0[1[2]]]}{[0[1][2]]}{2,1}

    \GPObject{aa}{(-8.75,-12.3)}
        {[0[1][2][3]]}{[0[1[2[3]]]]}{1,3}
    \GPObject{ab}{(-5.25,-12.3)}
        {[0[1][2[3]]]}{[0[1[2]][3]]}{2,1}
    \GPObject{ac}{(-1.75,-12.3)}
        {[0[1][2[3]]]}{[0[1[2][3]]]}{2,2}

    \GPObject{ba}{(1.75,-12.3)}
        {[0[1[2]][3]]}{[0[1][2[3]]]}{1,2}
    \GPObject{bb}{(5.25,-12.3)}
        {[0[1[2][3]]]}{[0[1][2[3]]]}{2,2}
    \GPObject{bc}{(8.75,-12.3)}
        {[0[1[2[3]]]]}{[0[1][2][3]]}{3,1}

    \draw[GPedge] (rootLabel.south) -- (a.north);
    \draw[GPedge] (rootLabel.south) -- (b.north);
    \draw[GPedge] (aLabel.south) -- (aa.north);
    \draw[GPedge] (aLabel.south) -- (ab.north);
    \draw[GPedge] (aLabel.south) -- (ac.north);
    \draw[GPedge] (bLabel.south) -- (ba.north);
    \draw[GPedge] (bLabel.south) -- (bb.north);
    \draw[GPedge] (bLabel.south) -- (bc.north);
\end{tikzpicture}
}
\caption{The first three levels of the generating tree for uncolored 2d-min and 2d-max heap pairs on $\Perm$ and $\Aneq$. The two heaps are shown side by side in the first two levels and vertically in the third level, with the label below.}
\label{fig:gen-tree-pair}
\end{figure}
\begin{figure}[H]
\centering
\forestset{
    heap tree/.style={
        for tree={
            circle,
            draw,
            line width=1pt,
            minimum size=13mm,
            inner sep=0pt,
            font=\Huge,
            edge={line width=1pt},
            l sep=6mm,
            s sep=2mm,
            anchor=center,
            child anchor=north,
            parent anchor=south
        }
    },
    blue node/.style={fill=blue!18},
    red node/.style={fill=red!18}
}

\resizebox{0.6\textwidth}{!}{%
\begin{tabular}{@{}c@{\hspace{18mm}}c@{}}
\begin{forest}
baseline=(current bounding box.north),
heap tree
[0, blue node
    [1, blue node]
    [2, red node
        [3, blue node]
        [4, blue node]
    ]
    [5, red node
        [6, blue node]
        [7, blue node
            [8, blue node]
        ]
        [9, red node
            [10, blue node]
        ]
    ]
]
\end{forest}
&
\begin{forest}
baseline=(current bounding box.north),
heap tree
[0, blue node
    [1, blue node
        [2, blue node]
    ]
    [3, red node]
    [4, blue node
        [5, blue node]
        [6, red node]
        [7, blue node]
    ]
    [8, red node
        [9, blue node]
        [10, red node]
    ]
]
\end{forest}
\\
\raisebox{-7mm}{{\Huge $\cMin(A)$}} &
\raisebox{-7mm}{{\Huge $\cMax(A)$}}
\end{tabular}%
}

\vspace{2mm}
{\small
\setlength{\arraycolsep}{4pt}
$
\begin{array}{|c!{\color{black!25}\vrule}c!{\color{black!25}\vrule}c!{\color{black!25}\vrule}c!{\color{black!25}\vrule}c!{\color{black!25}\vrule}c!{\color{black!25}\vrule}c!{\color{black!25}\vrule}c!{\color{black!25}\vrule}c!{\color{black!25}\vrule}c!{\color{black!25}\vrule}c!{\color{black!25}\vrule}c!{\color{black!25}\vrule}c|}
\hline
i & 0 & 1 & 2 & 3 & 4 & 5 & 6 & 7 & 8 & 9 & 10 & \\
\hline
A[i] &  & 4 & 2 & 5 & 5 & 1 & 3 & 3 & 6 & 2 & 4 &  \\
A' &  & 4 & 2 & 5 &  & 1 & 3 &  & 6 & 2 & 4 &  \\
\BP(\cMin(A')) & ( & ( & )( & ( &  & ))( & ( &  & ( & ))( & ( & )))) \\
\BP(\cMax(A')) & ( & ( & ( & ))( &  & ( & )( &  & ))( & ( & )( & ))) \\
U &  & 1 & 0 & 1 &  & 0 & 0 &  & 1 & 0 &  &  \\
S=s_2\cdots s_n &  &  & 0 & 10 &  & 10 & 0 &  & 10 & 10 & 0 &  \\
c_{\min} &  &  &  &  &  & 1 &  &  &  & 1 &  &  \\
c_{\max} &  &  &  & 1 &  &  &  &  & 1 &  &  &  \\
C &  & 0 & 0 & 0 & 1 & 0 & 0 & 1 & 0 & 0 & 0 &  \\
\hline
\end{array}
$
}
\vspace{-1mm}
\caption{An example of the sequences used in \Cref{subsec:ub-all} for an array $A\in\Agen_{10}$.}
\label{fig:cminmax}
\end{figure}

\clearpage
\endgroup

\section{Omitted Proofs in Section~\ref{sec:lb}}\label{app:proof-lb}

\subsection{Proof of Lemma~\ref{lem:quarter-walk}}
\restQuarterWalk*

The asymptotic behavior of quarter-plane walks is known from the work of Denisov et al.~\cite{denisov2015} and Bostan et al.~\cite{bostan2014}. Bouvel et al.~\cite{bouvel2018} used these results to derive the asymptotic behavior of strong-Baxter permutations. For our purposes, we use the formulation of these asymptotic results for step multisets given by Bostan et al.~\cite{bostan2021}. To state the lemma, we first introduce two necessary notations.

Given a step multiset \(\mathcal S\), let \(\Lambda\) be the lattice spanned by \(\mathcal S\). We say that \(\mathcal S\) is \emph{strongly aperiodic} if, for every lattice point \(u\in\Lambda\), the translated step multiset $u+\mathcal S=\{u+s:s\in\mathcal S\}$ also generates \(\Lambda\).

Note that every step multiset \(\mathcal S\) in Table~\ref{tab:gen-tree} is strongly aperiodic, since it contains $(0,0)$, $(1,0)$, and $(0,1)$.
Indeed, since \(\mathcal S\) contains \((1,0)\) and \((0,1)\), it generates \(\Lambda=\mathbb Z^2\). For every \(u\in\Lambda\), the translated step set \(u+\mathcal S\) contains \(u\), \(u+(1,0)\), and \(u+(0,1)\). Hence it generates their differences
\[
    u-u=(0,0),
    \quad
    (u+(1,0))-u=(1,0)
    \quad\text{and}\quad
    (u+(0,1))-u=(0,1),
\]
and therefore generates \(\mathbb Z^2=\Lambda\). Thus, \(\mathcal S\) is strongly aperiodic.

Given a step multiset \(\mathcal S\), we say that a point \((i,j)\in \mathbb Z_{\ge 0}^2\) is \emph{reachable from infinity} for \(\mathcal S\) if there exists a quarter-plane walk with steps in \(\mathcal S\) that starts at a point \((a,b)\) with $a>i$ and $b>j$, and ends at \((i,j)\).

Note that, for every step multiset \(\mathcal S\) in Table~\ref{tab:gen-tree}, every \((i,j)\in \mathbb Z_{\ge 0}^2\) is reachable from infinity, since \(\mathcal S\) contains $(-1,0)$ and $(0,-1)$.

\begin{lemma}[Theorem~3.11 in~\cite{bostan2021}]
Let $\mathcal S$ be a finite step multiset in $\mathbb Z^2$, with characteristic polynomial $Q_{\mathcal S}(x,y)$. Suppose that $\mathcal S$ is not contained in a half-plane, contains a step in $\mathbb Z_{\ge 0}^2$, and is strongly aperiodic. Then $Q_{\mathcal S}$ has a unique critical point $(a,b)\in\mathbb R_{>0}^2$.

Moreover, let $w_{n,i,j}$ denote the number of quarter-plane walks with steps in $\mathcal S$ of length $n$ that end at $(i,j)$. If $(i,j)$ is reachable from infinity for \(\mathcal S\), then, setting $\mu=Q_{\mathcal S}(a,b)$, we have
\[
    w_{n,i,j}=\Theta\!\left(\frac{\mu^n}{n^c}\right)
\]
for some constant $c>0$.
\end{lemma}

For any step multiset $\mathcal S$ satisfying the assumptions of Lemma~\ref{lem:quarter-walk}, the presence of $(1,0)$, $(-1,0)$, $(0,1)$, and $(0,-1)$ ensures that $\mathcal S$ is not contained in a half-plane.
The presence of $(0,0)$, $(1,0)$, and $(0,1)$ guarantees strong aperiodicity, and $(1,0)$ is a step in $\mathbb Z_{\ge0}^2$.
Moreover, $(0,0)$ is reachable from infinity using the steps $(-1,0)$ and $(0,-1)$.
Thus, applying the preceding lemma with $(i,j)=(0,0)$ and observing that $e_n=w_{n,0,0}$ proves Lemma~\ref{lem:quarter-walk}.

\section{Omitted Proofs in Section~\ref{sec:ub}}\label{app:proof-ub}
\subsection{Proof of Lemma~\ref{lem:ub-cmin}}\label{subsec:proof-ub-cmin}

\restUbCmin*

Let $N(n,k)$ denote the number of ordinal trees with $n+1$ nodes and $k$ leaves. These numbers are the Narayana numbers~\cite{riordan1975} and are given by
\[
    N(n,k)=\frac{1}{n}\binom{n}{k}\binom{n}{k-1}.
\]

\begin{lemma}\label{lem:narayana}
    Let $T$ be an ordinal tree with $n+1$ nodes and $k$ leaves. Then $T$ can be represented using $\log N(n,k)+o(n)$ bits while supporting $\wordacc$ on $\BP(T)$ in constant time.
\end{lemma}
\begin{proof}
    Remove the first and last parentheses of $\BP(T)$, corresponding to the
    root. Let the remaining balanced-parentheses sequence be $P$, of length
    $2n$. Then, we can write
    \[
        P =
        (^{a_1})^{d_1}
        (^{a_2})^{d_2}
        \cdots
        (^{a_k})^{d_k},
    \]
    where $a_1+\cdots+a_k=d_1+\cdots+d_k=n.$

    We first store a bit array $L[1,\dots, n]$ where $L[i]=1$ iff the $i$-th opening parenthesis of $P$ is followed by a closing parenthesis. Equivalently, $L[i]=1$ iff the $(i+1)$-th node in preorder is a leaf. Hence, $L$ has $k$ ones.

    We also store a bit array $R[1,\dots,n]$ where $R[i]=1$ iff the $i$-th closing parenthesis of $P$ is followed by an opening parenthesis. In this case, the node corresponding to this opening parenthesis is a non-leftmost child. Since an ordinal tree with $k$ leaves has $k-1$ non-leftmost child nodes, $R$ has $k-1$ ones.

    The two bit arrays are stored using
    \[
        \log\binom{n}{k}
        +
        \log\binom{n}{k-1}
        +o(n)
        =
        \log N(n,k)+o(n)
    \]
    bits by Lemma~\ref{lem:rrr}.

    It remains to support $\wordacc$ on $\BP(T)$. Partition $P$ into superblocks of length $\log^2 n$ and microblocks of length $b=(\log n)/4$. For each superblock, store the numbers of opening and closing parentheses preceding it. For each microblock, store the numbers of opening and closing parentheses preceding it within its superblock, together with the first parenthesis of the microblock. The total space is
    \[
        O\!\left(\frac{n}{\log^2 n}\log n\right)
        +
        O\!\left(\frac{n}{\log n}\log\log n\right) = o(n).
    \]
    
    To answer $\wordacc$ on $\BP(T)$, it suffices to reconstruct the constant number of microblocks intersecting the requested word. To decode a microblock, suppose that $o$ opening and $c$ closing parentheses precede it. We decode the next $b$ bits of $L$ starting from position $o+1$ and the next $b$ bits of $R$ starting from position $c+1$. This is supported in constant time since representations of $L$ and $R$ allow $\wordacc$ on them by Lemma~\ref{lem:rrr}. We use a precomputed table indexed by these two length-$b$ strings and the first parenthesis of the microblock; the table returns the whole microblock of $P$. This table is well defined because the next parenthesis is determined during the scan: after an opening parenthesis, the following parenthesis is closing exactly when the corresponding bit of $L$ is $1$, and after a closing parenthesis, the following parenthesis is opening exactly when the corresponding bit of $R$ is $1$. Since the table has size $2^{2b+1}\log n=o(n)$, any $O(\log n)$ consecutive bits of $P$, and hence of $\BP(T)$, can be accessed in constant time.
\end{proof}

Suppose $\cMin(A)$ has $k$ leaves. Then it has $k-1$ non-leftmost child nodes, and it is enough to store the colors of those nodes. Let $Y[1,\dots,k-1]$ be the array of their color bits, listed in preorder. We store $Y$ explicitly and encode the underlying tree $\Min(A)$ using Lemma~\ref{lem:narayana}. Thus, the total space is
\[
    \log N(n, k) + k + o(n)
\] bits. Let $\alpha=k/n$. Since
\[
    \log \binom{n}{k} \le nH(\alpha),
\]
where $H = \alpha \log {(1/\alpha)} + (1-\alpha) \log {(1/(1-\alpha))}$ is the binary entropy function, we have
\[
\begin{aligned}
    \log N(n,k)+k
    &= \log \binom{n}{k}
    + \log \binom{n}{k-1}
    - \log n
    + k \\
    &\le 2nH(\alpha)+\alpha n+o(n) \\
    &= n(2H(\alpha)+\alpha)+o(n).
\end{aligned}
\] The function $2H(\alpha)+\alpha$ is maximized at $\alpha=2-\sqrt{2}$, and its maximum value is $\log(3+2\sqrt{2})$. Therefore the space is at most $n\log(3+2\sqrt{2})+o(n)$ bits. 

By Lemma~\ref{lem:narayana}, $\wordacc$ on $\BP(\Min(A))$ can be supported in constant time. Thus, by Lemma~\ref{lem:cbp}, the required query support follows once $\wordacc$ on $\Xmin(A)$ is supported in constant time.

Since the nodes of $\Min(A)$ are labeled $0,\dots,n$ in preorder, with node $0$ as the root, the bit array $L[1,\dots,n]$ stored in Lemma~\ref{lem:narayana} satisfies $L[j]=1$ iff node $j$ is a leaf. Hence, for $i\ge 2$, node $i$ is a non-leftmost child iff $L[i-1]=1$.

The first two bits of $\Xmin(A)$ are always $0$, since they correspond to nodes $0$ and $1$, which are always blue. We therefore partition the suffix of $\Xmin(A)$ after these two bits into blocks of length $b=(\log n)/4$. Consider a block whose first entry stores the color of node $p$. We decode the next $b$ bits of $L$ starting from position $p-1$. Let $q=\rank_1(p-2,L)$, which is the number of non-leftmost child nodes before node $p$. We then read the next $b$ bits of $Y$ starting from position $q+1$. A precomputed table indexed by these two length-$b$ strings returns the corresponding block of $\Xmin(A)$: positions whose corresponding bit of $L$ is $0$ are filled with $0$, and the remaining positions are filled, in order, with the bits read from $Y$. The table has size $2^{2b}b=o(n)$ bits, so each block of $\Xmin(A)$ can be decoded in constant time. Consequently, $\wordacc$ on $\Xmin(A)$ is supported in constant time.

\subsection{Proof of Lemma~\ref{lem:cbp}}\label{subsec:cbp}

\restCbp*

First, define the following color-related operations on $\cMin(A)$:
\begin{itemize}
    \item $\clr(i)$: return the color of node $i$.
    \item $\PRS(i)$: return the rightmost red sibling to the left of node $i$.
    \item $\NRS(i)$: return the leftmost red sibling to the right of node $i$.
\end{itemize}

By \cite[Lemma~5]{jokim2025}, $\NSV$ and $\RMin$ on $A$ can be answered in constant time provided that the standard tree navigation operations (previous sibling, next sibling, subtree size, degree, level ancestor, child rank, child select, and parent)\footnote{See Table~1 of~\cite{navarro2014} for their definitions.} and the three color-related operations above are supported in constant time.

Moreover, $\PSV$ corresponds to the parent operation in $\Min(A)$. The standard tree navigation operations are supported on $\BP(\Min(A))$ in constant time using $o(n)$ additional bits~\cite{navarro2014}, and $\clr$ is supported by $\Xmin$ in constant time. Thus, it remains only to support $\PRS$ and $\NRS$.

To support these queries, we use the \emph{modified tree cover} introduced by Hamada et al.~\cite{hamada2024}. Let $T$ be an ordinal tree with $n$ nodes, and let $B \ge 1$ be a block-size parameter. The \emph{modified tree cover} is defined on a tree $T^{+}$ obtained from $T$ by inserting dummy nodes. Each dummy node $d$ is inserted as a child of a node $v$ of $T$, and a consecutive subsequence of the children of $v$ in $T$ is moved under $d$, preserving their order. The nodes of $T^{+}$ are partitioned into disjoint \emph{micro trees}, where each micro tree is a connected subgraph of $T^{+}$.

Contracting each micro tree $\mu$ into a single node $\tau(\mu)$ gives the \emph{top-tier tree} $T_{\mathrm{top}}$. For two micro trees $\mu$ and $\nu$, we say that $\nu$ is a child micro tree of $\mu$ if $\tau(\nu)$ is a child of $\tau(\mu)$ in $T_{\mathrm{top}}$.

The \emph{modified tree cover} satisfies the following properties:
\begin{enumerate}
    \item The number of micro trees is $O(n/B)$.
    \item Each micro tree is of one of the following two types:
        \begin{enumerate}
        \item A \emph{singleton micro tree}, which consists of a single non-dummy node and may have an arbitrary number of child micro trees.
        \item A \emph{non-singleton micro tree}, which consists of fewer than $2B$ nodes and has at most one child micro tree. Its root $d$ is a dummy node, and all remaining nodes are non-dummy nodes.
        \end{enumerate}
\end{enumerate}

They showed that, for each micro tree $\mu$, $\BP(\mu)$ can be obtained from at most two substrings of $\BP(T)$ by concatenating them and, if $\mu$ is non-singleton, enclosing the resulting string with the parentheses corresponding to the dummy root.

The same observation applies to the color sequence. Let $X(\mu)$ denote the subsequence of $\Xmin(A)$ consisting of the colors of the non-dummy nodes of $\mu$ in preorder. If the non-dummy nodes of $\mu$ correspond to one substring of $\BP(\Min(A))$, then their preorder numbers form a consecutive interval, and hence $X(\mu)$ is obtained from the corresponding substring of $\Xmin(A)$. If they correspond to two substrings of $\BP(\Min(A))$, then their preorder numbers form at most two consecutive intervals, and $X(\mu)$ is obtained by concatenating the two corresponding substrings of $\Xmin(A)$.

We set $B=(\log n)/8$. Using the approach of Hamada et al.~\cite{hamada2024}, $\BP(\mu)$ can be obtained in constant time from $\wordacc$ on $\BP(T)$ for any micro tree $\mu$. Similarly, $X(\mu)$ can be obtained in constant time from $\wordacc$ on $\Xmin(A)$. Thus, the pair consisting of $\BP(\mu)$ and $X(\mu)$ can be used as an index to a universal precomputed table with at most $2^{6B}$ entries.

With this table, the remaining support for $\PRS$ and $\NRS$ is obtained by adapting the tree-cover-based support for $\NRS$ in Fischer's representation of colored 2d-min heaps~\cite{fischer2011}. The main difference is that the top-tier tree $T_{\mathrm{top}}$ of the modified tree cover replaces the mini-tree layer used in Fischer's construction.

\subsection{Proof of Lemma~\ref{lem:all-wordacc}}\label{subsec:proof-all-wordacc}

\restAllWordacc*

We use the following formulations of the reconstruction results of Jo and Kim~\cite{jokim2025}.
Their Theorem~2 shows how to support $\wordacc$ on both $\BP(\Min(A))$ and $\BP(\Max(A))$ using $U$, an alternative representation of $S$, and $o(n)$ bits of auxiliary structures.
Specifically, their formulation uses two sequences, denoted $D$ and $E$ in their paper, to represent $S$ in compressed form.
The same reconstruction framework applies when constant-time $\wordacc$ on $S$ is available directly, yielding the first lemma below.
The second lemma formulates their Theorem~3 in terms of the bit array $C$ and the induced subsequence $A'$.

\begin{lemma}[{\cite[Theorem~2]{jokim2025}}]\label{lem:cminmax}
    Given an array $A\in \Aneq$, suppose there is a data structure that supports $\wordacc$ on both $U$ and $S$ in constant time. Then, with $o(n)$ bits of auxiliary structures, $\wordacc$ on both $\BP(\Min(A))$ and $\BP(\Max(A))$ can be supported in constant time.
\end{lemma}

\begin{lemma}[{\cite[Theorem~3]{jokim2025}}]\label{lem:gen-bp}
    Given an array $A\in\Agen_n$, consider the bit array $C$ and the induced subsequence $A'\in\Aneq_m$. Suppose there is a data structure that supports $\wordacc$ on both $\BP(\Min(A'))$ and $\BP(\Max(A'))$ in constant time. Then, together with $C$ and $o(n)$ additional bits, $\wordacc$ on both $\BP(\Min(A))$ and $\BP(\Max(A))$ can be supported in constant time.
\end{lemma}

\subparagraph{Query support for $\Aneq_n$.}
By Lemmas~\ref{lem:cbp} and~\ref{lem:cminmax}, it remains only to support $\wordacc$ on both $\Xmin(A)$ and $\Xmax(A)$.

In the present representation, $U$ is exactly the leaf bit array for $\Min(A)$, while the complement of $U$ is the leaf bit array for $\Max(A)$ by Proposition~\ref{prop:minmax}. Thus, for $i\ge2$, node $i$ is a non-leftmost child in $\Min(A)$ iff $U[i-1]=1$, and it is a non-leftmost child in $\Max(A)$ iff $U[i-1]=0$.

Recall that we regard $S$ as part of the balanced-parentheses sequence $\BP(T)$ of an ordinal tree $T$, and store $T$ using Lemma~\ref{lem:narayana}. The representation of Lemma~\ref{lem:narayana} uses the bit array $L$. Since $L[j]=1$ iff the $(j+1)$-st vertex of $T$ in preorder is a leaf, the discussion above implies that, for $i\ge3$, if node $i$ is a non-leftmost child in $\Min(A)$, then its left sibling is internal iff $L[i-2]=1$.

Hence, for $\cMin(A)$, the color of node $i$ is determined as follows. If $i=0$ or $i=1$, then it is always blue. If $i=2$, then it is blue when $U[1]=0$ and red when $U[1]=1$. For $i\ge3$, the color is stored in $c_{\min}$ precisely when $U[i-1]=1$ and $L[i-2]=1$. If $U[i-1]=1$ and $L[i-2]=0$, then the color is forced to be red; if $U[i-1]=0$, then node $i$ is a leftmost child in $\Min(A)$, so its color is blue. The colors in $\cMax(A)$ are determined symmetrically.

We handle the colors of nodes $0$, $1$, and $2$ separately and partition the remaining suffix of $\Xmin(A)$ into blocks of length $b=(\log n)/4$.
We group these blocks into superblocks of length $O(\log^2 n)$.
For each superblock, we store the number of entries of $c_{\min}$ corresponding to nodes preceding it, and for each block, we store the corresponding count within its superblock.
These counts use $o(n)$ bits in total and determine the starting position in $c_{\min}$ for each block in constant time.

Consider a block whose first entry stores the color of node $p$.
We read the next $b$ bits of $U$ starting from position $p-1$, the next $b$ bits of $L$ starting from position $p-2$, and the next $b$ bits of $c_{\min}$ starting from the position determined by the stored counts.
A precomputed table indexed by these three strings reconstructs the block according to the color rules above.
The precomputed table has size $2^{3b}b=o(n)$ bits.
Since any requested word intersects only a constant number of blocks, both color bit sequences support $\wordacc$ in constant time using $o(n)$ additional bits.

\subparagraph{Query support for $\Agen_n$.}

For query support, by Lemma~\ref{lem:gen-bp}, it suffices to support $\wordacc$ on both $\Xmin(A)$ and $\Xmax(A)$ in constant time. Consider the color bit sequences $\Xmin(A')$ and $\Xmax(A')$ for $A'$. For each $i\in[1,n]$, let $j=\rank_0(i,C)$. If $C[i]=0$, then node $i$ of $A$ corresponds to node $j$ of $A'$, and its colors in $\cMin(A)$ and $\cMax(A)$ are given by the corresponding entries of $\Xmin(A')$ and $\Xmax(A')$, respectively. If $C[i]=1$, then by Proposition~\ref{prop:cminmax}, node $i$ is forced to be blue in both colored heaps.

The first bit of $\Xmin(A)$ is always $0$, since it corresponds to the root. We therefore partition the suffix of $\Xmin(A)$ after this bit into blocks of length $b=(\log n)/4$. Consider a block whose first entry stores the color of node $p$. Let $j=\rank_0(p-1,C)+1$. We decode the next $b$ bits of $C$ starting from position $p$, and the next $b$ bits of $\Xmin(A')$ starting from the entry for node $j$. A precomputed table indexed by these two length-$b$ strings returns the corresponding block of $\Xmin(A)$: positions whose corresponding bit of $C$ is $1$ are filled with $0$, and the remaining positions are filled, in order, with the bits read from $\Xmin(A')$. The table has size $2^{2b}b=o(n)$ bits, so each block of $\Xmin(A)$ can be decoded in constant time. The same construction supports $\wordacc$ on $\Xmax(A)$. Consequently, by Lemma~\ref{lem:cbp}, all desired queries are supported in constant time.

\section{Lower bounds for query subfamilies}\label{app:lb}

\begin{table}[ht]
\centering
\resizebox{\textwidth}{!}{
\begin{tblr}{
    colspec={|c|l|l|c|c|c|c|},
    cells={valign=m,halign=c},
    column{2,3}={halign=l},
}
\hline
\SetCell[c=2]{c} Generating tree &
&
\SetCell[r=2]{c} Step multiset $\mathcal S$ &
\SetCell[r=2]{c} {Positive critical\\point $(a,b)$} &
\SetCell[r=2]{c} {Growth\\constant} &
\SetCell[c=2]{c} Data structure &
\\
\cline{1-2}\cline{6-7}
Root & \SetCell{c} Rule & & & & Queries & Array type \\
\hline
\SetCell[r=4]{c} $(1)$ &
\SetCell[r=2]{l} $(k)\rightsquigarrow(1)(2)\cdots(k+1)$ &
\SetCell[r=2]{c} -- &
\SetCell[r=2]{c} -- &
\SetCell[r=2]{c} 4 &
$\PSV$ & $\Perm, \Aneq, \Agen$ \\
\cline{6-7}
& & & & & $\PSV, \NSV$ & $\Perm$ \\
\cline{2-7}
&
$(k)\rightsquigarrow(1)^2(2)^2\cdots(k-1)^2(k)(k+1)$ &
\SetCell{c} -- &
-- &
$2+2\sqrt{2}$ &
$\PSV,\NSV$ & $\Aneq$ \\
\cline{2-7}
&
$(k)\rightsquigarrow(1)^2(2)^2\cdots(k)^2(k+1)$ &
\SetCell{c} -- &
-- &
$3+2\sqrt{2}$ &
$\PSV,\NSV$ & $\Agen$ \\
\hline
\SetCell[r=9]{c} $(1,1)$ &
\SetCell[r=4]{l} {$
    (k,\ell)\rightsquigarrow
    \begin{cases}
        (1,\ell+1)(2,\ell+1)\cdots(k,\ell+1),\\
        (k+1,1)(k+1,2)\cdots(k+1,\ell)
    \end{cases}
$} &
\SetCell[r=4]{l} {$
    \begin{aligned}
    \{&(1,0),(0,1),(-1,0),(0,-1),\\
       &(1,-1),(-1,1),2(0,0)\}
    \end{aligned}
$} &
\SetCell[r=4]{c} $(1,1)$ &
\SetCell[r=4]{c} $8$ &
$\PSV,\PLV$ & $\Perm, \Aneq$ \\
\cline{6-7}
& & & & & $\PSV,\NLV$ & $\Perm, \Aneq, \Agen$ \\
\cline{6-7}
& & & & & $\PSV,\NSV,\PLV$ & $\Perm$ \\
\cline{6-7}
& & & & & {$\PSV,\NSV,$\\$\PLV,\NLV$} & $\Perm$ \\
\cline{2-7}
&
{$
    (k,\ell)\rightsquigarrow
    \begin{cases}
        (1,\ell+1)(2,\ell+1)\cdots(k,\ell+1),\\
        (k+1,1)(k+1,2)\cdots(k+1,\ell),\\
        (k,\ell)
    \end{cases}
$} &
{$
    \begin{aligned}
    \{&(1,0),(0,1),(-1,0),(0,-1),\\
       &(1,-1),(-1,1),3(0,0)\}
    \end{aligned}
$} &
$(1,1)$ &
$9$ &
$\PSV,\PLV$ & $\Agen$ \\
\cline{2-7}
&
{$
    (k,\ell)\rightsquigarrow
    \begin{cases}
        (1,\ell+1)^2(2,\ell+1)^2\cdots(k-1,\ell+1)^2(k,\ell+1),\\
        (k+1,1)(k+1,2)\cdots(k+1,\ell)
    \end{cases}
$} &
{$
    \begin{aligned}
    \{&(1,0),(0,1),2(-1,0),(0,-1),\\
       &(1,-1),2(-1,1),2(0,0)\}
    \end{aligned}
$} &
$(\sqrt{2},1)$ &
$4+4\sqrt{2}$ &
$\PSV,\NSV,\PLV$ & $\Aneq$ \\
\cline{2-7}
&
{$
    (k,\ell)\rightsquigarrow
    \begin{cases}
        (1,\ell+1)^2(2,\ell+1)^2\cdots(k-1,\ell+1)^2(k,\ell+1),\\
        (k+1,1)(k+1,2)\cdots(k+1,\ell),\\
        (k,\ell)
    \end{cases}
$} &
{$
    \begin{aligned}
    \{&(1,0),(0,1),2(-1,0),(0,-1),\\
       &(1,-1),2(-1,1),3(0,0)\}
    \end{aligned}
$} &
$(\sqrt{2},1)$ &
$5+4\sqrt{2}$ &
$\PSV,\NSV,\PLV$ & $\Agen$ \\
\cline{2-7}
&
{$
    (k,\ell)\rightsquigarrow
    \begin{cases}
        (1,\ell+1)^2(2,\ell+1)^2\cdots(k-1,\ell+1)^2(k,\ell+1),\\
        (k+1,1)^2(k+1,2)^2\cdots(k+1,\ell-1)^2(k+1,\ell)
    \end{cases}
$} &
{$
    \begin{aligned}
    \{&(1,0),(0,1),2(-1,0),2(0,-1),\\
       &2(1,-1),2(-1,1),2(0,0)\}
    \end{aligned}
$} &
$(\sqrt{2},\sqrt{2})$ &
$6+4\sqrt{2}$ &
{$\PSV,\NSV,$\\$\PLV,\NLV$} & $\Aneq$ \\
\cline{2-7}
&
{$
    (k,\ell)\rightsquigarrow
    \begin{cases}
        (1,\ell+1)^2(2,\ell+1)^2\cdots(k-1,\ell+1)^2(k,\ell+1),\\
        (k+1,1)^2(k+1,2)^2\cdots(k+1,\ell-1)^2(k+1,\ell),\\
        (k,\ell)
    \end{cases}
$} &
{$
    \begin{aligned}
    \{&(1,0),(0,1),2(-1,0),2(0,-1),\\
       &2(1,-1),2(-1,1),3(0,0)\}
    \end{aligned}
$} &
$(\sqrt{2},\sqrt{2})$ &
$7+4\sqrt{2}$ &
{$\PSV,\NSV,$\\$\PLV,\NLV$} & $\Agen$ \\
\hline
\end{tblr}
}
\caption{Generating trees and their corresponding quarter-plane step multisets, where applicable.}
\label{tab:gen-tree}
\end{table}

In this section, we derive ITLBs for each nonempty proper subfamily of $\{\PSV,\PLV,\NSV,\NLV\}$ on the three array classes $\Perm$, $\Aneq$, and $\Agen$, using the framework developed in Section~\ref{sec:lb}.
For subfamilies containing three queries, symmetry allows us to restrict our analysis to $\{\PSV,\PLV,\NSV\}$.
Table~\ref{tab:gen-tree} summarizes the succession rules, the corresponding step multisets and growth constants where applicable, and the query families and array classes to which they apply.

\subsection{The $\{\PSV,\PLV\}$ case.}

We follow the framework of Section~\ref{sec:lb}, counting realizable pairs of uncolored 2d-min and 2d-max heaps.
Rule~\eqref{rule:minmax-no-cons} is known to generate the Baxter numbers, yielding an ITLB of $3n-\Theta(\log n)$ bits for $\Perm$ and $\Aneq$, matching the lower bound obtained by Gawrychowski and Nicholson~\cite{gawrychowski2015}.
Thus, it remains to analyze the case of $\Agen$.
Applying the same framework as in Section~\ref{subsec:gen-to-lb} to rule~\eqref{rule:minmax}, we find that the number of realizable heap pairs of size $n$ equals the number of quarter-plane excursions of length $n-1$ with step multiset
\[
    \mathcal S=\{(1,0),(0,1),(-1,0),(0,-1),(1,-1),(-1,1),3(0,0)\}.
\]
Its characteristic polynomial is
\[
    Q_{\mathcal S}(a,b)=a+b+\frac{1}{a}+\frac{1}{b}+\frac{a}{b}+\frac{b}{a}+3.
\]
The unique positive critical point is $(1,1)$, at which $Q_{\mathcal S}(1,1)=9$.
Applying Lemma~\ref{lem:quarter-walk} and taking logarithms of the resulting counts yields the following lower bound.

\begin{theorem}\label{thm:lb-minmax}
    For arrays in $\Agen_n$, the ITLB for the simultaneous support of $\PSV$ and $\PLV$ is $n\log 9-\Theta(\log n) \approx 3.17n-\Theta(\log n)$ bits.
\end{theorem}

\subsection{The $\{\PSV,\NSV\}$ case.}
For $\Perm$, the colors of a colored 2d-min heap are completely determined by its topology, so counting colored heaps is equivalent to counting uncolored heaps, yielding an ITLB of $2n-\Theta(\log n)$ bits.
For $\Agen$, Fischer~\cite{fischer2011} already established an ITLB of $n\log(3+2\sqrt{2})-\Theta(\log n) \approx 2.54n-\Theta(\log n)$ bits.
Thus, it remains to analyze the case of $\Aneq$.

We count distinct colored 2d-min heaps induced by arrays $A\in \Aneq_n$.
We use the same labeling from Section~\ref{subsec:lb-min}.
As in Section~\ref{subsec:lb-min}, the underlying 2d-min heap is determined by the number of indices popped from the stack at each insertion.
Thus, as in rule~\eqref{rule:min}, each node labeled $k$ has one child with each label in $\{1,2,\dots,k+1\}$ before accounting for colors.
It remains to account for the color of the newly inserted node.

Consider the step of the construction algorithm that processes $A[m+1]$.
If no pop occurs, then the new node is appended as the leftmost child of node $m$.
Hence its color is forced to be blue.
This case gives a child node labeled $k+1$.

If exactly one pop occurs, node $m$ becomes a left leaf sibling of node $m+1$.
Since $A[m]\neq A[m+1]$, the color of the new node is forced to be red, as noted in Proposition~\ref{prop:cminmax}.
This case gives a child node labeled $k$.

If more than one pop occurs, then the new node has a left non-leaf sibling.
In this case, by choosing $A[m+1]$ to be equal to or different from the value of that sibling, either blue or red can be realized for the new node.
Thus, each node labeled $k$ has two children with each label in $\{1,2,\dots,k-1\}$, one for each possible color, and only one child with each of the labels $k$ and $k+1$.

Thus, the generating tree for colored 2d-min heaps on arrays with no consecutive equal elements is described by
\begin{equation}\label{rule:cmin-no-cons}
    \text{Root: }(1) \qquad
    \text{Rule: }(k) \rightsquigarrow (1)^2(2)^2\cdots(k-1)^2(k)(k+1).
\end{equation}
We next estimate the number of objects generated by rule~\eqref{rule:cmin-no-cons}, which yields the desired lower bound.

\begin{theorem}\label{thm:cmin-no-cons-lb}
    Let $a_n$ be the number of objects generated by rule~\eqref{rule:cmin-no-cons}. Then,
    \[
        a_n = \Theta\left(\frac{(2+2\sqrt{2})^n}{n^{1.5}}\right).
    \]
    Hence, the ITLB for the simultaneous support of  $\PSV$ and $\NSV$ on arrays with no consecutive equal elements is $n \log (2+2\sqrt{2}) - \Theta(\log n) \approx 2.27n - \Theta(\log n)$ bits.
\end{theorem}
\begin{proof}
    Let $f_{n,k}$ be the number of nodes of size $n$ with label $k$, and set
    \[
        \mathcal F(x;y)=\sum_{n\ge1}\sum_{k\ge1} f_{n,k}x^ny^k .
    \]
    The root contributes $xy$. For a node with label $k$, the
    rule~\eqref{rule:cmin-no-cons} contributes
    \[
        2(y+\cdots+y^{k-1})+y^k+y^{k+1}.
    \]
    Therefore,
    \begin{align*}
        \mathcal F(x;y)
        &=xy+x\sum_{n\ge1}\sum_{k\ge1}f_{n,k}x^n
        \left(2(y+\cdots+y^{k-1})+y^k+y^{k+1}\right)\\
        &=xy+x\sum_{n\ge1}\sum_{k\ge1}f_{n,k}x^n
        \left(2\frac{y-y^k}{1-y}+y^k+y^{k+1}\right)\\
        &=xy+\frac{2x}{1-y}
        \left(
            y\sum_{n\ge1}\sum_{k\ge1}f_{n,k}x^n
            -
            \sum_{n\ge1}\sum_{k\ge1}f_{n,k}x^ny^k
        \right)\\
        &\qquad
        +x\sum_{n\ge1}\sum_{k\ge1}f_{n,k}x^n(y^k+y^{k+1})\\
        &=xy+\frac{2x}{1-y}\bigl(y\mathcal F(x;1)-\mathcal F(x;y)\bigr)
            +x(1+y)\mathcal F(x;y).
    \end{align*}
    Multiplying by $1-y$ and rearranging gives
    \begin{equation}\label{eq:cmin-kernel}
    \bigl(1-y+x+xy^2\bigr)\mathcal F(x;y)=xy(1-y)+2xy\mathcal F(x;1).
    \end{equation}
    Here, our goal is to determine the generating function
    \[
        \mathcal F(x;1)=\sum_{n\ge1} a_nx^n .
    \]
    We therefore choose $y=Y(x)$ so that the left-hand side of \eqref{eq:cmin-kernel} becomes zero and leaves an equation for $\mathcal F(x;1)$. Hence $Y$ must satisfy
    \[
        1-Y+x+xY^2=0.
    \]
    Solving this equation for $Y$ gives the two roots
    \[
        \frac{1\pm\sqrt{1-4x-4x^2}}{2x}.
    \]
    Since $Y$ will be substituted into the generating function, we need $Y$ to be a power series in $x$. Thus we choose
    \[
        Y=\frac{1-\sqrt{1-4x-4x^2}}{2x}.
    \]
    Substituting $y=Y$ into \eqref{eq:cmin-kernel} gives
    \[
        0=xY(1-Y)+2xY\mathcal F(x;1),
    \]
    and hence
    \[
        \mathcal F(x;1)= \frac{1-2x-\sqrt{1-4x-4x^2}}{4x}.
    \]

    This generating function, and the corresponding sequence $(a_n)_{n\ge1}$, are listed in~\cite{oeis} as \href{https://oeis.org/A071356}{A071356}. The desired asymptotic follows by standard singularity analysis on the generating function~\cite[Chapter VI]{flajolet2009}; it is also noted in~\cite[\href{https://oeis.org/A071356}{A071356}]{oeis}.
\end{proof}

\subsection{The $\{\PSV,\NLV\}$ case.}

Let $A^R$ denote the reversal of $A$, defined by $A^R[j] = A[n-j+1]\text{ for } j\in[1,n].$ For an index $i$ in $A$, we write $\rev(i) := n - i + 1$ for the corresponding index in $A^R$. Under this correspondence, $\PSV$ and $\PLV$ on $A^R$ correspond to $\NSV$ and $\NLV$ on $A$, respectively. Similarly, $\RMin(\cdot,\cdot,n)$ and $\RMax(\cdot,\cdot,n)$ on $A^R$ correspond to $\RMin(\cdot,\cdot,1)$ and $\RMax(\cdot,\cdot,1)$ on $A$, respectively. We now relate the constructions of $\Min(A)$ and $\Max(A^R)$.

\begin{proposition} \label{prop:min-reversed-max}
For every $1\le i<n$, no pop occurs at step $i+1$ in the construction of $\Min(A)$ if and only if no pop occurs at step $\rev(i)$ in the construction of $\Max(A^R)$. Consequently, node $i+1$ is a leftmost child in $\Min(A)$ if and only if node $\rev(i)$ is a leftmost child in $\Max(A^R)$.
\end{proposition}

\begin{proof}
Immediately before step $i+1$ in the construction of $\Min(A)$, the top of the stack is $i$. Therefore, no pop occurs at this step if and only if $A[i] < A[i+1]$. In the construction of $\Max(A^R)$, immediately before step $\rev(i)$, the top of the stack is $\rev(i+1)$. Thus, no pop occurs at this step if and only if $A^R[\rev(i+1)] > A^R[\rev(i)]$. By the definition of $A^R$, this condition is equivalent to  $A[i+1] > A[i]$. Hence, the two conditions are equivalent.
\end{proof}

Here, we first restrict the class of arrays to permutations. Thus, it suffices to count the realizable pairs $(\Min(P),\Max(P^R))$ induced by arrays $P\in \Perm_n$.  We use a pair of labels $(k,\ell)$. The first label $k$ is defined as in the $\PSV$ case for the construction of $\Min(P)$. The second label $\ell$ is the number of children of the root in $\Max(P^R)$.

Consider the step of the construction algorithm that processes a new element $x=P[m+1]$, viewed as appending $x$ as the $(m+1)$-st element of $P$. Since $P$ is a permutation, $x$ represents a new rank inserted among the existing values; in $P^R$, this corresponds to prepending the same new element.

Let $T_1,T_2,\ldots,T_\ell$ be the subtrees rooted at the children of the root of $\Max(P^R)$, from left to right, and let $w_1 < w_2 < \cdots < w_\ell$ be the values at their roots.
After $x$ is prepended to $P^R$, it is the first element processed by the max-heap construction and has no preceding larger element, so its node is attached as the leftmost child of the root.
It has $T_1,T_2,\ldots,T_t$ as its child subtrees for some $0\le t\le\ell$.
Moreover, every value of $t$ is realizable: for $t=0$, choose $x < w_1=P[m]$; for $1\le t\le\ell$, choose $x$ so that $w_t< x < w_{t+1}$, where $w_{\ell+1}=+\infty$. If $t=0$, the new root degree becomes $\ell+1$. If $t \ge 1$, then the new root degree becomes $1+(\ell-t)$. Thus, any value in $1,2,\ldots,\ell$ can be the second label.

By Proposition~\ref{prop:min-reversed-max}, the case $t\ge 1$ occurs if and only if no pop occurs in the construction of $\Min(P)$ while processing $x=P[m+1]$. Moreover, the choices of $x$ required by the two labels are compatible. The case $t=0$ is exactly the range $x\le P[m]$, where a pop occurs in $\Min(P)$ and the first label can be chosen from $1,\ldots,k$ as in the $\PSV$ case. The cases $t\ge1$ all lie in the range $x>P[m]$, where no pop occurs and the first label becomes $k+1$. Therefore we obtain the same succession rule as rule~\eqref{rule:minmax-no-cons}, which is the Baxter succession rule. Thus, we obtain the following lower bound.

For arrays in $\Aneq_n$ and $\Agen_n$, allowing repeated values introduces no additional children beyond those prescribed by the Baxter succession rule.
It can only make the sequence $w_1,w_2,\ldots,w_\ell$ nondecreasing rather than strictly increasing, and hence may remove some choices of $t$ for a particular array.
Nevertheless, the second label of each child is still one of the values $1,2,\ldots,\ell+1$ described above, with the first label determined as before.
Since permutations already realize all pairs counted by the Baxter succession rule, the number of realizable pairs $(\Min(A),\Max(A^R))$ is the Baxter number in all three array classes.
Thus, we obtain the following lower bound.

\begin{theorem}\label{thm:minrmax-lb}
    For permutations, arrays with no consecutive equal elements and general arrays, the ITLB for the simultaneous support of $\PSV$ and $\NLV$ is $3n-\Theta(\log n)$ bits.
\end{theorem}

\subsection{The $\{\PSV,\PLV,\NSV\}$ case.}

\begin{theorem}\label{thm:three-lb}
    For arrays in $\Aneq_n$, the ITLB for the simultaneous support of $\PSV$,$\PLV$, and $\NSV$ is $n\log(4+4\sqrt{2})-\Theta(\log n)
    \approx 3.27n-\Theta(\log n)$ bits.
    For arrays in $\Agen_n$, it is $n\log(5+4\sqrt{2})-\Theta(\log n) \approx 3.41n-\Theta(\log n)$ bits.
\end{theorem}
\begin{proof}
    For arrays in $\Aneq_n$, it is enough to count realizable pairs consisting of a colored 2d-min heap and an uncolored 2d-max heap. The same labeling as in Section~\ref{subsec:lb-cminmax} gives the following succession rule:
    \[
        \text{Root: }(1,1) \qquad
        (k,\ell)\rightsquigarrow
        \begin{cases}
            (1,\ell+1)^2(2,\ell+1)^2\cdots(k-1,\ell+1)^2(k,\ell+1),\\
            (k+1,1)(k+1,2)\cdots(k+1,\ell).
        \end{cases}
    \]
    This rule gives the step multiset
    \[
     \mathcal S=\{(1,0),(0,1),2(-1,0),(0,-1),
           (1,-1),2(-1,1),2(0,0)\},
    \]
    whose characteristic polynomial is
    \[
        Q_{\mathcal S}(x,y)
        =x+y+\frac{2}{x}+\frac{1}{y}
         +\frac{x}{y}+\frac{2y}{x}+2.
    \]
    A direct calculation shows that the positive critical point is $(\sqrt{2},1)$. At this point,
    \[
        Q_{\mathcal S}(\sqrt{2},1)=4+4\sqrt{2}.
    \]

    For arrays in $\Agen_n$, adding one copy of the step $(0,0)$ changes the characteristic polynomial from $Q_{\mathcal S}(x,y)$ to $Q_{\mathcal S}(x,y)+1$. The positive critical point remains $(\sqrt{2},1)$, and at this point the new value is $(4+4\sqrt{2})+1=5+4\sqrt{2}$.
\end{proof}

\section{Upper bounds for query subfamilies}\label{app:ub}

\subsection{The $\{\PSV,\PLV\}$ case.}

For $\Perm_n$ and $\Aneq_n$, storing $U$ and $S$ and applying Lemma~\ref{lem:cminmax} supports the desired queries in constant time using $3n+o(n)$ bits, matching the upper bounds of Gawrychowski and Nicholson~\cite{gawrychowski2015}.

For $A\in\Agen_n$, we use the same bit array $C$ and induced subsequence $A'\in\Aneq_m$ as in Section~\ref{subsec:ub-all}.
We store the above representation for $A'$, together with $C$ encoded using Lemma~\ref{lem:rrr}.
The total space is
\[
    3m+\log\binom{n}{m}+o(n)
    \le n\log 9+o(n)
    \approx 3.17n+o(n)
\]
bits.
By Lemma~\ref{lem:gen-bp}, this supports $\wordacc$ on both $\BP(\Min(A))$ and $\BP(\Max(A))$ in constant time.
The desired queries are then supported through the corresponding tree navigation operations, yielding the following theorem.

\begin{theorem}\label{thm:ub-minmax}
    Given an array $A\in\Agen_n$, there is a data structure using $n\log 9+o(n) \approx 3.17n+o(n)$ bits that supports $\PSV$, $\PLV$, $\RMin(\cdot,\cdot,n)$, and $\RMax(\cdot,\cdot,n)$ on $A$ in constant time.
\end{theorem}

\subsection{The $\{\PSV,\NSV\}$ case.}

For $\Perm_n$, the colors are determined by the heap topology, so an uncolored 2d-min heap representation suffices to support the desired queries in constant time using $2n+o(n)$ bits~\cite{fischerheun2011,fischer2011}.
For $\Agen_n$, Corollary~\ref{cor:ub-cmin} supports the desired queries in constant time using $n\log(3+2\sqrt{2})+o(n)$ bits, matching Fischer's upper bound~\cite{fischer2011}.
Thus, we describe only the case of $\Aneq_n$.

For arrays in $\Aneq_n$, Proposition~\ref{prop:cminmax} implies that the colors of non-leftmost child nodes whose left sibling is a leaf are forced. We exploit this property to improve the space bound.

\begin{theorem}\label{thm:ub-cmin-no-cons}
    Given an array $A\in\Aneq_n$, there is a data structure using $n\log(2+2\sqrt{2})+o(n) \approx 2.27n+o(n)$ bits that supports $\PSV$, $\NSV$, and $\RMin$ on $A$, as well as $\wordacc$ on both $\BP(\Min(A))$ and $\Xmin(A)$, in constant time.
\end{theorem}
\begin{proof}
    Similar to the proof of Lemma~\ref{lem:ub-cmin} (Appendix~\ref{subsec:proof-ub-cmin}), we encode $\BP(\Min(A))$ and the colors separately. Let $k$ be the number of leaves of $\Min(A)$, and let $t$ be the number of non-leftmost child nodes whose left sibling is an internal node. We describe below an encoding of $\BP(\Min(A))$ using
    \begin{equation}\label{eq:ktbp}
        \log \binom{n}{k}
        + \log \binom{k-1}{t}
        + \log \binom{n-k}{t}
        + o(n)
    \end{equation}
    bits that supports $\wordacc$ on $\BP(\Min(A))$ in constant time.

    By Proposition~\ref{prop:cminmax}, the only colors that are not forced are those of the $t$ non-leftmost child nodes whose left sibling is internal. Let $Y[1,\dots,t]$ be the array of their color bits, listed in preorder. We store $Y$ explicitly, and hence the total space is
    \[
        \log \binom{n}{k}
        + \log \binom{k-1}{t}
        + \log \binom{n-k}{t}
        + t
        + o(n)
    \]
    bits.

    Let $\alpha=k/n$ and $\beta=t/n$. Then, the total space is at most
    \[
        n\left(
            H(\alpha)
            +\alpha H(\beta/\alpha)
            +(1-\alpha)H(\beta/(1-\alpha))
            +\beta
        \right)
        +o(n)
    \]
    bits. The expression in parentheses is maximized at $\alpha=1/2$ and $\beta=(2-\sqrt{2})/2$ and its maximum value is $\log(2+2\sqrt{2})$. Therefore, the total space is at most
    \[
         n\log(2+2\sqrt{2})+o(n)
    \] bits.

    By Lemma~\ref{lem:cbp}, the query support follows once we support $\wordacc$ on both $\BP(\Min(A))$ and $\Xmin(A)$ in constant time. We first describe a representation of $\BP(\Min(A))$ that supports $\wordacc$ in constant time; then we show how to support $\wordacc$ on $\Xmin(A)$.

    \subparagraph{Supporting $\wordacc$ on $\BP(\Min(A))$.}
    Following the proof of Lemma~\ref{lem:narayana}, remove the first and last parentheses of $\BP(\Min(A))$, corresponding to the root, and let $P$ be the remaining sequence. Use the same bit array $L[1,\dots,n]$, where $L[i]=1$ iff the $i$-th opening parenthesis of $P$ is followed by a closing parenthesis. Then $L$ has $k$ ones and can be stored in $\log\binom{n}{k}+o(n)$ bits.

    We encode the boundaries of the maximal runs of closing parentheses in $P$ using two bit arrays $R_1$ and $R_2$ as follows. Excluding the final closing run of $P$, the number of maximal closing runs of length at least two is exactly $t$.

    We construct a bit array $R_1[1,\dots,k-1]$ over the leaves, except for the last one, in postorder: $R_1[i]=1$ iff, for the corresponding node $v$, the closing parenthesis of $v$ is followed by a closing parenthesis. Since the closing parenthesis of a leaf is preceded by an opening parenthesis, $R_1$ marks the starting points of the closing runs of length at least two. Hence, $R_1$ has $t$ ones and can be stored in $\log\binom{k-1}{t}+o(n)$ bits.

    In addition, we construct a bit array $R_2[1,\dots,n-k]$ over the non-root internal nodes in postorder: $R_2[i]=1$ iff, for the corresponding node $v$, the closing parenthesis of $v$ is followed by an opening parenthesis. Since the closing parenthesis of an internal node is preceded by a closing parenthesis, $R_2$ marks the ending points of the maximal closing runs of length at least two, excluding the final run. Hence, $R_2$ has $t$ ones and can be stored in $\log\binom{n-k}{t}+o(n)$ bits.

    Thus the total space is \eqref{eq:ktbp}. It remains to support $\wordacc$ on $\BP(\Min(A))$. We use the same two-level blocking scheme for $P$ as in the proof of Lemma~\ref{lem:narayana}, and decode each microblock using a precomputed table. For each superblock, we store the following three quantities: (a) number of opening parentheses, (b) number of closing parentheses of leaves preceding it, and (c) number of closing parentheses of internal nodes preceding it.
    For each microblock, we store the same three quantities preceding it within its superblock, together with its first parenthesis and one additional bit indicating whether the preceding parenthesis is opening.

    To decode a microblock, suppose that $o$ opening parentheses, $c_1$ closing parentheses of leaves, and $c_2$ closing parentheses of internal nodes precede it. We decode the next $b$ bits of $L$ starting from position $o+1$, the next $b$ bits of $R_1$ starting from position $c_1+1$, and the next $b$ bits of $R_2$ starting from position $c_2+1$.

    Given these three length-$b$ strings, the first parenthesis of the microblock, and the additional bit, a precomputed table reconstructs the whole microblock of $P$.
    After an opening parenthesis, the next parenthesis is closing exactly when the corresponding bit of $L$ is $1$.
    After a closing parenthesis preceded by an opening parenthesis, the next parenthesis is closing exactly when the corresponding bit of $R_1$ is $1$.
    After a closing parenthesis preceded by another closing parenthesis, the next parenthesis is opening exactly when the corresponding bit of $R_2$ is $1$.
    If the microblock starts with a closing parenthesis, the additional bit determines which of the latter two rules applies.
    Since the table has size $2^{3b+2}\log n=o(n)$, any $O(\log n)$ consecutive bits of $P$, and hence of $\BP(\Min(A))$, can be accessed in constant time.

    \subparagraph{Supporting $\wordacc$ on $\Xmin(A)$.}
    Since the nodes of $\Min(A)$ are labeled $0,\dots,n$ in preorder, with node $0$ as the root, the leaf bit array $L[1,\dots,n]$ satisfies $L[j]=1$ iff node $j$ is a leaf. Hence, for $i\ge 2$, node $i$ is a non-leftmost child iff $L[i-1]=1$. Among these nodes, the ones whose left sibling is internal are determined by $R_1$.

    We use the same block decoding scheme as in the proof of Lemma~\ref{lem:all-wordacc} (Appendix~\ref{subsec:proof-all-wordacc}), with $L$, $R_1$, and $Y$ determining the colors here.
    The first two bits of $\Xmin(A)$ are always $0$, since they correspond to nodes $0$ and $1$, which are always blue.
    We partition the remaining suffix into blocks of length $b=(\log n)/4$.
    Consider a block whose first entry stores the color of node $p$. Let $q=\rank_1(p-2,L)$, which is the number of non-leftmost child nodes before node $p$, and let $r=\rank_1(q,R_1)$, which is the number of such nodes whose left sibling is internal. We decode the next $b$ bits of $L$ starting from position $p-1$, the next $b$ bits of $R_1$ starting from position $q+1$, and the next $b$ bits of $Y$ starting from position $r+1$. A precomputed table indexed by these three length-$b$ strings returns the corresponding block of $\Xmin(A)$: positions whose corresponding bit of $L$ is $0$ are filled with $0$, positions whose corresponding bit of $L$ is $1$ and whose corresponding bit of $R_1$ is $0$ are filled with $1$, and the remaining positions are filled, in order, with the bits read from $Y$. The table has size $2^{3b}b=o(n)$ bits, so each block of $\Xmin(A)$ can be decoded in constant time. Consequently, $\wordacc$ on $\Xmin(A)$ is supported in constant time.
    \end{proof}

\subsection{The $\{\PSV,\NLV\}$ case.}\label{subsec:ub-minrmax}

\begin{figure}[t]
\centering
\forestset{
    heap tree/.style={
        for tree={
            circle,
            draw,
            line width=1pt,
            minimum size=13mm,
            inner sep=0pt,
            font=\Huge,
            edge={line width=1pt},
            l sep=6mm,
            s sep=2mm,
            anchor=center,
            child anchor=north,
            parent anchor=south
        }
    }
}

\resizebox{0.52\textwidth}{!}{
\begin{tabular}{@{}c@{\hspace{18mm}}c@{}}
\begin{forest}
baseline=(current bounding box.north),
heap tree
[0
    [1]
    [2
        [3]
        [4]
    ]
    [5
        [6]
        [7
            [8]
        ]
        [9
            [10]
        ]
    ]
]
\end{forest}
&
\begin{forest}
baseline=(current bounding box.north),
heap tree
[11
    [10
        [9]
    ]
    [8
        [7]
        [6
            [5]
        ]
        [4]
        [3
            [2]
            [1]
        ]
    ]
]
\end{forest}
\\
\raisebox{-7mm}{{\Huge $\Min(A)$}} &
\raisebox{-7mm}{{\Huge $\Max(A^R)$}}
\end{tabular}
}

\vspace{2mm}
{\small
\setlength{\arraycolsep}{4pt}
$
\begin{array}{|c!{\color{black!25}\vrule}c!{\color{black!25}\vrule}c!{\color{black!25}\vrule}c!{\color{black!25}\vrule}c!{\color{black!25}\vrule}c!{\color{black!25}\vrule}c!{\color{black!25}\vrule}c!{\color{black!25}\vrule}c!{\color{black!25}\vrule}c!{\color{black!25}\vrule}c!{\color{black!25}\vrule}c!{\color{black!25}\vrule}c|}
\hline
i &  & 1 & 2 & 3 & 4 & 5 & 6 & 7 & 8 & 9 & 10 & \\
\hline
A[i] &  & 4 & 2 & 5 & 5 & 1 & 3 & 3 & 6 & 2 & 4 & \\
A^R[i] &  & 4 & 2 & 6 & 3 & 3 & 1 & 5 & 5 & 2 & 4 & \\
U &  & 0 & 1 & 0 & 1 & 1 & 0 & 1 & 0 & 1 & 0 & \\
\BP(\Min(A)) & (( & )(( & )( & ) & )(( & )(( & ) & )(( & ) & ) & ) & ) \\
R_{\min} &  & 1 & 1 & 0 & 1 & 1 & 0 & 1 & 0 & 0 & 0 & \\
\BP(\Max(A^R)) & ((( & ) & )(( & )(( & ) & )( & )(( & )( & ) & ) & ) & ) \\
R_{\max} &  & 0 & 1 & 1 & 0 & 1 & 1 & 1 & 0 & 0 & 0 & \\
\hline
\end{array}
$
}
\vspace{-1mm}
\caption{An example of the sequences used in Appendix~\ref{subsec:ub-minrmax} for an array $A\in\Agen_{10}$. For readability, each node in $\Max(A^R)$ is labeled with its corresponding index $i$ in $A$ rather than with the index $\rev(i)$ in $A^R$. For the rows $\BP(\Min(A))$ and $\BP(\Max(A^R))$, the entry in column $i$ contains the $i$-th closing parenthesis followed by the subsequent run of opening parentheses. Thus $R_{\min}[i]$ and $R_{\max}[i]$ are $1$ exactly when the corresponding entry contains an opening parenthesis after its first symbol.
}
\label{fig:minrmax}
\end{figure}
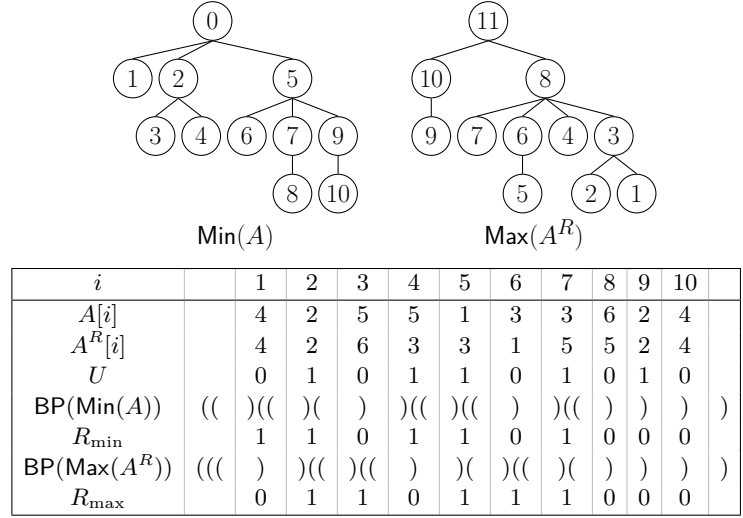

In this subsection, we give a representation for the pair $\Min(A)$ and $\Max(A^R)$. Recall that $\rev(i)=n-i+1$ is the index in $A^R$ corresponding to index $i$ in $A$. This pair supports $\PSV$ and $\RMin(\cdot,\cdot,n)$ from $\Min(A)$, and also supports $\NLV$ and $\RMax(\cdot,\cdot,1)$ from $\Max(A^R)$ through this correspondence. The remaining symmetric case uses the pair $\Min(A^R)$ and $\Max(A)$, and is handled in the same way.

\begin{theorem}\label{thm:ub-minrmax}
    Given an array $A\in\Agen_n$ (or $A\in\Aneq_n$), there is a data structure using $3n+o(n)$ bits that supports $\PSV$, $\NLV$, $\RMin(\cdot,\cdot,n)$, and $\RMax(\cdot,\cdot,1)$ on $A$ in constant time.
\end{theorem}
\begin{proof}
    The desired queries are obtained from standard navigation on $\Min(A)$ and $\Max(A^R)$. The queries $\PSV$ and $\NLV$ are parent queries in $\Min(A)$ and $\Max(A^R)$, respectively, under the correspondence $i\leftrightarrow\rev(i)$. Moreover, $\RMin(\cdot,\cdot,n)$ and $\RMax(\cdot,\cdot,1)$ are answered using a constant number of standard navigation operations on $\Min(A)$ and $\Max(A^R)$, respectively~\cite{fischerheun2011}; these operations are supported from balanced-parentheses sequences with $o(n)$ additional bits~\cite{navarro2014}. Thus, it remains to show that $\wordacc$ on both $\BP(\Min(A))$ and $\BP(\Max(A^R))$ can be supported in constant time.
    
    Store a bit array $U[1,\dots,n]$ marking the non-leftmost child nodes of $\Min(A)$. By Proposition~\ref{prop:min-reversed-max}, for every $j\ge2$, node $j$ is a non-leftmost child in $\Max(A^R)$ iff $U[\rev(j)+1]=1$. Thus, the non-leftmost child nodes for $\Max(A^R)$ are also determined by $U$. Also store the $R_{\min}$ and $R_{\max}$, which are $R$ arrays defined in the proof of Lemma~\ref{lem:narayana} for $\Min(A)$ and $\Max(A^R)$, respectively. These three arrays use $3n$ bits. See Figure~\ref{fig:minrmax} for an example.
    
    In the proof of Lemma~\ref{lem:narayana}, word access to the balanced-parentheses sequence is supported using the leaf array $L$ together with the array $R$. Since, for every $i\ge2$, node $i$ is a non-leftmost child iff node $i-1$ is a leaf, the array $U$ also determines the leaf arrays $L$ for both $\Min(A)$ and $\Max(A^R)$. Together with the explicitly stored arrays $R_{\min}$ and $R_{\max}$, the decoding method of Lemma~\ref{lem:narayana} therefore supports $\wordacc$ on both $\BP(\Min(A))$ and $\BP(\Max(A^R))$ in constant time with $o(n)$ additional bits.
\end{proof}

\subsection{The $\{\PSV,\PLV,\NSV\}$ case.}

In this subsection, we show how to combine a representation of one colored heap with an uncolored representation of the other heap. We present the construction for $\cMin(A)$ and $\Max(A)$, which gives the query family $\{\PSV,\PLV,\NSV, \RMin, \RMax(\cdot,\cdot,n)\}$; the remaining cases follow by exchanging min and max, and by applying the construction to the reversed array $A^R$. We first handle arrays in $\Aneq_n$, and then extend the construction to arrays in $\Agen_n$.

\begin{theorem}\label{thm:ub-three}
    Given an array $A\in\Aneq_n$, there is a data structure using $n\log(4+4\sqrt{2})+o(n)$ bits that supports $\PSV$, $\PLV$, $\NSV$, $\RMin$, and $\RMax(\cdot,\cdot,n)$ on $A$ in constant time. For arrays $A\in\Agen_n$, the same queries can be supported in constant time using $n\log(5+4\sqrt{2})+o(n)$ bits.
\end{theorem}

\begin{proof}
Let $A\in\Aneq_n$. Our data structure stores the representation of $\cMin(A)$ from Theorem~\ref{thm:ub-cmin-no-cons}. Hence, the queries $\PSV$, $\NSV$, and $\RMin$ are supported in constant time. It remains to support the information needed from $\Max(A)$. For this, we store the bit array $R$ for $\Max(A)$, defined as in Lemma~\ref{lem:narayana}, using $n$ bits.
The representation of $\cMin(A)$ already contains the bit array $L$, where $L[j]=1$ iff node $j$ is a leaf in $\Min(A)$.
By Proposition~\ref{prop:minmax}, $L$ also determines the leaf information for $\Max(A)$: for $1\le j<n$, node $j$ is a leaf iff $L[j]=0$, and node $n$ is always a leaf.
Thus, applying the decoding method from the proof of Lemma~\ref{lem:narayana} using this information and the stored array $R$ supports $\wordacc$ on $\BP(\Max(A))$ in constant time.

The query $\PLV$ is given by the parent operation in $\Max(A)$, and $\RMax(\cdot,\cdot,n)$ can be answered using a constant number of standard navigation operations on $\Max(A)$~\cite{fischerheun2011}. Since these operations are supported from $\BP(\Max(A))$ with $o(n)$ additional bits~\cite{navarro2014}, the required queries involving $\Max(A)$ are supported in constant time. The space usage for $A\in\Aneq_n$ is
\[
    n\log(2+2\sqrt{2}) + n + o(n)
    = n\log(4+4\sqrt{2})+o(n)
\]
bits.

For an array $A\in\Agen_n$, we use the same reduction as in Section~\ref{subsec:ub-all}.
Define the bit array $C$ and the induced subsequence $A'\in\Aneq_m$ as in that subsection.
We store the above representation for $A'$, together with $C$ encoded using Lemma~\ref{lem:rrr}.
By the same space calculation as in Section~\ref{subsec:ub-all}, this uses
\[
    m\log(4+4\sqrt{2})+\log\binom{n}{m}+o(n)
    \le n\log(5+4\sqrt{2})+o(n)
\]
bits.
By Lemma~\ref{lem:gen-bp}, this supports $\wordacc$ on both $\BP(\Min(A))$ and $\BP(\Max(A))$ in constant time.
Moreover, we support $\wordacc$ on $\Xmin(A)$ in constant time using the block decoding method for general arrays described in the proof of Lemma~\ref{lem:all-wordacc} (Appendix~\ref{subsec:proof-all-wordacc}).
Consequently, Lemma~\ref{lem:cbp} and the standard tree navigation operations support all desired queries in constant time.
\end{proof}

\end{document}